\documentclass[%
 reprint,
 amsmath,
 amssymb,
 aps,
 superscriptaddress,
 pra,
 nofootinbib
]{revtex4-2}

\usepackage{graphicx}% Include figure files
\usepackage{dcolumn}% Align table columns on decimal point
\usepackage{bm}% bold math

\usepackage{array}[=2016-10-06]% centered paragraph cols in tables
\newcolumntype{P}[1]{>{\centering\arraybackslash}p{#1}}  % centered paragraph column

\usepackage[dvipsnames]{xcolor}
\usepackage{hyperref}
\usepackage{comment}
\hypersetup{
  breaklinks=true,
  colorlinks=true,
  allcolors=BlueViolet,
}

\usepackage{braket}
\newcommand{\E}{\mathbb{E}}  % expectation value

\DeclareMathOperator{\poly}{poly}

\usepackage{amsthm}
\newtheorem{theorem}{Theorem}[section]
\newtheorem{definition}[theorem]{Definition}
\usepackage{thm-restate}

\usepackage[caption=false]{subfig}
\usepackage[boxruled]{algorithm2e}
\SetAlgorithmName{Protocol}{protocols}{List of Protocols}
\SetKwInput{kwSetup}{Setup}
\SetKwInput{kwOutput}{Output}

\usepackage{physics}
\usepackage{quantikz}
\newcommand{\InfleqtionM}{Infleqtion, Madison, WI, 53703}

\begin{document}

\preprint{APS/123-QED}

\title{Supercheq: Quantum Advantage for Distributed Databases}

\author{Eric R. Anschuetz$^{*\dagger}$}
\affiliation{Infleqtion, Chicago, IL, USA}
\affiliation{Department of Physics, Massachusetts Institute of Technology}

\author{Pranav Gokhale$^*$}
\affiliation{Infleqtion, Chicago, IL, USA}

\author{Behnam Tonekaboni}
\affiliation{Infleqtion, Melbourne, VIC, Australia}

\author{Colin Campbell}
\affiliation{Infleqtion, Chicago, IL, USA}

\author{Frederic T. Chong}
\affiliation{Infleqtion, Chicago, IL, USA}
\affiliation{Department of Computer Science, University of Chicago}

\author{Edward D. Dahl$^\dagger$}
\affiliation{Infleqtion, Chicago, IL, USA}

\author{Paige Frederick$^\dagger$}
\affiliation{Infleqtion, Chicago, IL, USA}

\author{Eric B. Jones}
\affiliation{Infleqtion, Chicago, IL, USA}

\author{Benjamin Hall}
\affiliation{Infleqtion, Chicago, IL, USA}

\author{Salahedeen Issa$^\dagger$}
\affiliation{Infleqtion, Chicago, IL, USA}

\author{Palash Goiporia}
\affiliation{Infleqtion, Chicago, IL, USA}

\author{Junyu Liu}
\affiliation{Department of Computer Science, University of Chicago}
\affiliation{Pritzker School of Molecular Engineering, University of Chicago}
\affiliation{Kadanoff Center for Theoretical Physics, University of Chicago}
\affiliation{Department of Computer Science, University of Pittsburgh}

\author{Stephanie Lee}
\affiliation{Infleqtion, Chicago, IL, USA}

\author{Peter Noell}
\affiliation{Infleqtion, Chicago, IL, USA}

\author{Victory Omole}
\affiliation{Infleqtion, Chicago, IL, USA}

\author{David Owusu-Antwi}
\affiliation{Infleqtion, Chicago, IL, USA}

\author{Michael A. Perlin$^\dagger$}
\affiliation{Infleqtion, Chicago, IL, USA}

\author{Rich Rines}
\affiliation{Infleqtion, Chicago, IL, USA}

\author{Mark Saffman}
\affiliation{\InfleqtionM}
\affiliation{Department of Physics, University of Wisconsin-Madison, Madison, WI, 53706}

\author{Kaitlin N. Smith}
\affiliation{Infleqtion, Chicago, IL, USA}
\affiliation{Department of Computer Science, Northwestern University}

\author{Teague Tomesh}
\affiliation{Infleqtion, Chicago, IL, USA}
\affiliation{Department of Computer Science, Princeton University}

\footnote[0]{$^*$ These two authors contributed equally.}

\footnote[0]{$^\dagger$ E.R.A. is now at Caltech, E.D.D is now at IonQ, P.F. will be joining UC Berkeley in Fall 2026, S.I. is now at Princeton University, and M.A.P. is now at JPMorganChase.}

\date{\today}

\begin{abstract}
We introduce Supercheq, a family of quantum protocols that achieves asymptotic advantage over classical protocols for checking the equivalence of files, a task also known as fingerprinting. The first variant, Supercheq-EE (Efficient Encoding), uses $n$ qubits to verify files with $2^{O(n)}$ bits---an exponential advantage in communication complexity (i.e.~bandwidth, often the limiting factor in networked applications) over the best possible classical protocol in the simultaneous message passing setting. Moreover, Supercheq-EE can be gracefully scaled down for implementation on circuits with $\poly(n^\ell)$ depth to enable verification for files with $O(n^\ell)$ bits for arbitrary constant $\ell$. The quantum advantage is achieved by random circuit sampling, thereby potentially endowing circuits from recent quantum supremacy and quantum volume experiments with a practical application. We validate Supercheq-EE's performance at scale through GPU simulation motivated by Infleqtion's Sqale neutral atom QPU gateset. The second variant, Supercheq-IE (Incremental Encoding), also achieves arbitrary-polynomial advantage in fingerprint size ($n$ qubits to verify files with size $O(n^{\ell})$ bits), while supporting incremental updates to the fingerprint using only a constant number of $(\ell-1)$-qubit gates. Moreover, Supercheq-IE at $\ell=2$ ($\geq 3$) only requires Clifford gates (gates in the $\ell-1$ level of the Clifford hierarchy), ensuring relatively modest overheads for error-corrected implementation. We experimentally demonstrate proof-of-concepts on quantum hardware from Diraq (spin qubit) and IBM (superconducting). We envision Supercheq could be deployed in distributed data settings, accompanying replicas of important databases.

\end{abstract}

\maketitle

\section{Introduction}

A hallmark of modern networks and databases is distributed replication of files, whether to improve availability, redundancy, or performance. However, replication requires protocols for integrity verification to check that copies of data have not diverged. This motivates the task of \textit{fingerprinting}: mapping an input bitstring to a shorter bitstring (the fingerprint) in order to distinguish two inputs with (arbitrarily) high success probability.

Fingerprinting has been extensively studied in the so-called \emph{simultaneous message passing} (SMP) setting~\cite{yao1979some, newman1996public, ambainis1996communication}. In this model, fingerprints of two bitstrings $A$ and $B$ are sent to a referee to verify whether $A=B$ with high probability. The goal is to minimize the communication complexity, i.e.~the number of bits transferred to the referee, of this process. Previous work~\cite{buhrman2001quantum} has shown that there is an exponential advantage in communication complexity when utilizing quantum states for this task. Namely, for $A$ and $B$ each of $N$ bits, there is a quantum protocol that succeeds with high probability using $O\left(\log N \right)$ qubits, while there is a known achievable lower-bound of $\Theta\left(\sqrt{N}\right)$ classical bits in performing this task~\cite{newman1996public}. Unfortunately, the protocol of \cite{buhrman2001quantum} requires constructing a complicated superposition state that generally takes time scaling polynomially with $N$ to prepare on a quantum computer.

\begin{figure}
    \centering
    \includegraphics[width=0.5\textwidth]{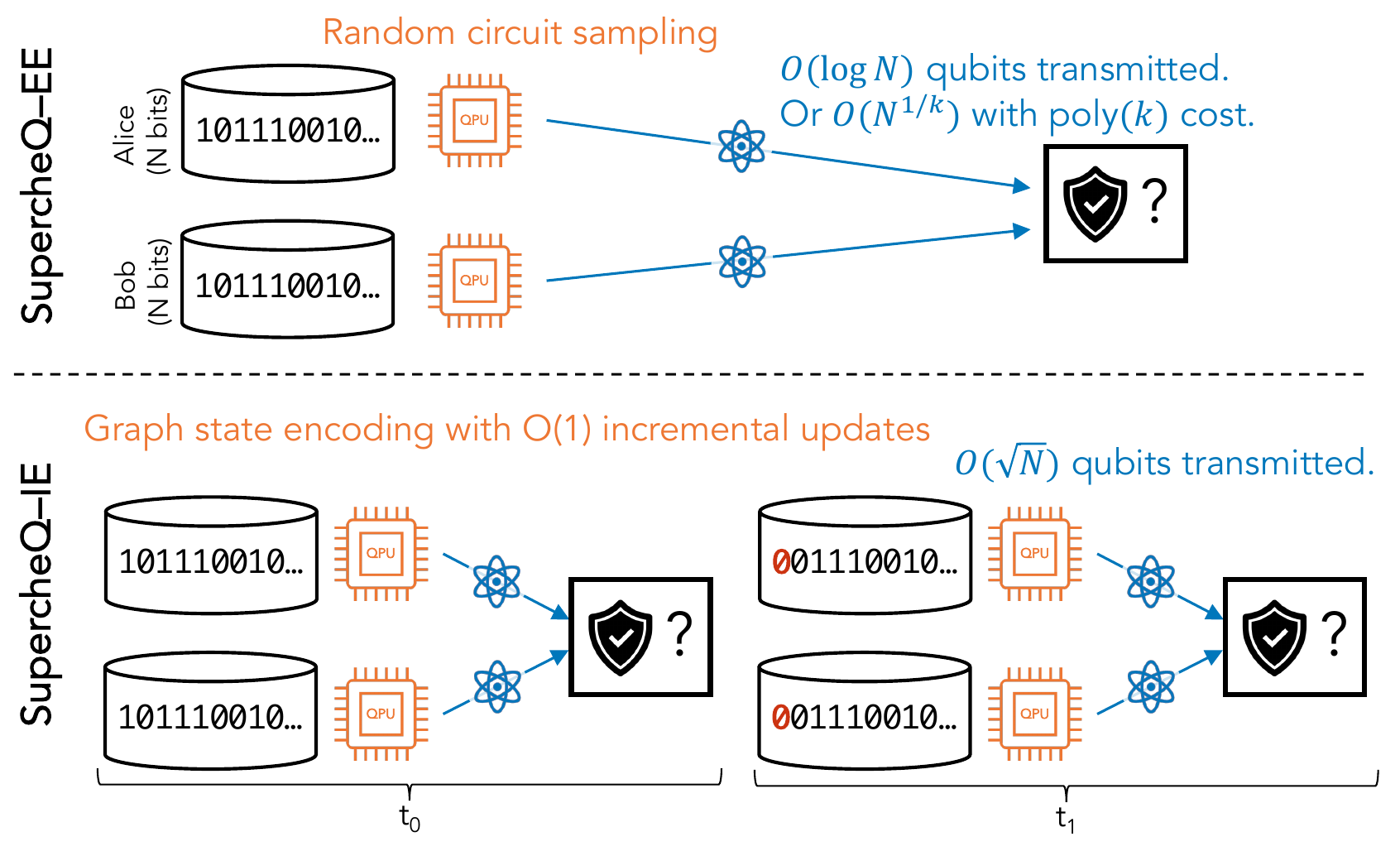}
    \caption{The EE (Efficient Encoding) variant of Supercheq compares two $N$-bit files by transmitting as few as $O(\log N)$ qubits---an exponential advantage over classical protocols in the SMP setting. The IE (Incremental Encoding) graph/hypergraph variants match/exceed the best-possible $\Theta(\sqrt{N})$ scaling of classical protocols while achieving constant-cost updates for incremental file changes.}
    \label{fig:overview}
\end{figure}

\begin{table*}
\footnotesize
\begin{tabular}{p{1.94cm}||P{1.99cm}|P{1.77cm}|P{1.62cm}|P{1.87cm}|P{1.31cm}|P{1.15cm}|P{1.4cm}|P{2.55cm}|P{0.65cm}}
Protocol  & Hash & Quantum FP & \multicolumn{3}{c|}{\textbf{Supercheq-EE} {[\S \ref{subsec:haar_random}, \ref{subsec:approx_unitary_t_design}, \ref{subsec:hw_eff}]}}  & \multicolumn{2}{c|}{\textbf{Supercheq-IE} {[\S\ref{sec:Supercheq_ie}]}} & Classical FP & Naive \\
{[Example]} & {[SHA-256]} & \cite{buhrman2001quantum} & {[Haar Rdm}] & {[Aprx $t$-Dsgn}] & {[HW-Eff]} & Graph & Hypergraph & \cite{ambainis1996communication} & \\ \hline \hline
FP Size $(n)$ & $O(1)$, [256 bits] &  $O(\log{N})$ & $O(\log{N})$ & $O(N^{1/\ell})$ & Empirical & $\left[\sim \sqrt{2N} \right]$ & $O(N^{1/\ell})$ & $\Theta \left(\sqrt{N} \right), \left[\sim \sqrt{3N} \right]$ & $N$ \\
Quantum Cost & N/A & $\exp(n)$ & $\exp(n)$ & $\poly(n^\ell)$ & Empirical & $O(n^2)$ & $O(n^\ell)$ & N/A & N/A \\
Collision-Free? & \textcolor{red}{\textbf{No}} & Yes & Yes & Yes & Yes & Yes & Yes & Yes & Yes \\
Incremental? & ?, {[No]} & No & No & No & No & \textcolor{teal}{\textbf{Yes}} & \textcolor{teal}{\textbf{Yes}} &  No & \textcolor{teal}{\textbf{Yes}} \\
\end{tabular}
\caption{Comparison of fingerprinting (FP) protocols for files with $N$ classical bits, in ascending order of the fingerprint size. We report our results in the special case of $\ell$ held constant with respect to $N$ for simplicity, though similar bounds hold in the more general case. Quantum cost refers to the gate count needed to produce the quantum fingerprint state. Collision-free refers to whether the protocol is with high probability safe against worst-case choices of files; in particular, hashing to a constant size (e.g. SHA-256) is not collision-free because there always exist (many) pairs of files that map to the same hashed value. Incremental refers to whether fingerprints can be updated with constant cost when a constant number of bits in the underlying file are flipped. We introduce Supercheq-Efficient Encoding (EE), which encompasses three subcategories: Haar-Random (\S \ref{subsec:haar_random}), Approximate Unitary $t$-Design (\S \ref{subsec:approx_unitary_t_design}), and Hardware-Efficient variant (\S \ref{subsec:hw_eff}). We also introduce Supercheq-Incremental Encoding (IE), which is the only protocol here that achieves incremental updates.}
\label{tab:comparison}
\end{table*}

To address these difficulties we introduce Supercheq, a family of quantum protocols that exhibit asymptotic advantage over classical protocols for fingerprinting, as summarized in Fig.~\ref{fig:overview} and Table~\ref{tab:comparison}. The first variant, Supercheq-EE (Efficient Encoding), achieves an (up-to) exponential advantage in communication complexity. This is often the limiting factor for modern distributed databases. Supercheq-EE can be used to check if two $2^{O(n)}$-bit files are identical by sending only $n$ qubits, which is an exponential advantage relative to the best-possible classical protocol. In a nearer-term setting with restricted quantum circuit gate counts and connectivity, Supercheq-EE can compare files of size $O(n^\ell)$ bits for arbitrarily large constant $\ell$, with quantum circuit cost that scales as $\poly(n^\ell)$. For comparison, $\ell=2$ is known to be the best that can be achieved classically in the SMP setting \cite{newman1996public, kremer1999randomized}. Interestingly, Supercheq-EE's core mechanism is random circuit sampling. As such, Supercheq-EE endows circuits from quantum supremacy and quantum volume experiments with the first plausible application for communication to our knowledge (see \cite{hangleiter2022computational}), though scalable utilization of Supercheq will likely require fault-tolerant quantum computing with error correction.

The second variant, Supercheq-IE (Incremental Encoding), has two subvariants---graph and hypergraph---that at least match the $\ell=2$ quadratic scaling of the best possible classical protocol, but with the advantage of being \textit{incremental}. For instance, in the Supercheq-IE graph subvariant, if a bit is flipped in an $N$-bit source file, the $n = \Theta(\sqrt{N})$-qubit fingerprint can be updated by applying a single gate. To our knowledge, no nontrivial classical fingerprinting protocols guarantee this property, nor are they likely to achieve this property. The Supercheq-IE protocol has two somewhat counterintuitive features:
\begin{enumerate}
    \item For the graph subvariant, the essential ingredient is a Clifford circuit, which is known to be efficiently classically simulable with a quadratic memory overhead. However, this quadratic separation is known to be optimal~\cite{karanjai2018contextuality,PhysRevX.12.021037,anschuetz2022interpretable} (i.e. it takes quadratically more classical bits than qubits to simulate Clifford circuits). As a result, quantum advantages are known to still persist in a variety of communication complexity settings, as recently demonstrated by \cite{PhysRevX.12.021037,anschuetz2022interpretable}.
    \item The protocol does not involve Grover search and no \emph{quantum random access memory} (QRAM) is involved---only an ordinary classical database.
\end{enumerate}

These features are appealing for error-corrected implementations of Supercheq-IE. Since only Clifford circuits are required for the graph subvariant, Supercheq-IE avoids the overhead of magic state distillation for non-Clifford operations~\cite{ogorman2017quantum}. In the hypergraph subvariant with degree-$\ell$ hyperedges, only gates in the $\ell - 1$ level of Clifford hierarchy are required, enabling cheap error corrected implementations with codes such as the hypercube quantum code \cite{kubica2015unfolding, hangleiter2025fault}. Moreover, in both subvariants, no special hardware for classical data loading is required.

The remainder of this paper is organized as follows. Sec.~\ref{sec:background} provides background on fingerprinting for checking equivalence of data. Secs.~\ref{sec:Supercheq-ee} and~\ref{sec:Supercheq_ie} specify the Supercheq-EE (Efficient Encoding) and Supercheq-IE (Incremental Encoding) protocols, respectively. Since our asymptotic bounds for Supercheq-EE are loose, we employ GPU-accelerated simulation motivated by Infleqtion's Sqale neutral atom QPU gateset to elucidate the real-world advantage achievable with Supercheq in Sec.~\ref{sec:simulation}. Next, we demonstrate experimental results for Supercheq, executed on quantum hardware from Diraq (spin qubit) and IBM (superconducting), in Sec.~\ref{sec:experimental}. Finally, we conclude in Sec.~\ref{sec:conclusion}. Appendix~\ref{app:proofs} presents mathematical proofs for key results, and Appendix~\ref{app:noisy} contains additional plots pertaining to noisy simulation of Supercheq-EE.

\section{Background} \label{sec:background}

\subsection{Notation and Setup}
Throughout this paper, $N$ will denote the number of classical bits composing a file of interest, while $n$ will denote the number of qubits (or bits) in the fingerprint. We focus exclusively on the simultaneous message passing (SMP) setting~\cite{yao1979some}. Under this model, we consider the communication complexity for Alice and Bob to send the fingerprint to a third-party referee, as depicted in Fig.~\ref{fig:overview}. The SMP setting allows for both Alice and Bob to have private coins (i.e., uncorrelated randomness), but assumes that Alice and Bob do not share a secure random key (which would require a secure key exchange protocol).

\subsection{Fingerprinting}
The most naive classical procedure for fingerprinting a file is to send the entire file itself. In this case---which corresponds to the right-most column in Table~\ref{tab:comparison}---we have that $n = N$. On the other end of the spectrum, we may also consider a hash function such as SHA-256, which maps every input file to a fingerprint (message digest) of $n = 256$ bits. In many cases, a hash comparison will suffice for checking equality. However, hashing has a worst-case error of 100\% in the sense that there will always exist differing files that collide with the same fingerprint (hash value). This is especially concerning in a security-sensitive or adversarial setting---for example, when Alice and Bob's files originate from a malicious supplier \cite{scott2005optimal}.

Restricting ourselves to the collision-free setting, where we approach zero worst-case error with high probability (w.h.p.), even classically it is possible to beat the $n = N$ naive solution. In particular, \cite{ambainis1996communication} demonstrated a classical protocol that uses private coin flips to generate a fingerprint of size $\sim \sqrt{3N}$ bits that achieves a worst case error of at most $5/11 < 0.5$. This error can be suppressed exponentially by sending $M$ independent fingerprints to reduce the worst-case error to $\left(5/11\right)^M$. This quadratic improvement over naive fingerprinting is in fact optimal classically~\cite{newman1996public, kremer1999randomized}, i.e.~the best one can do classically is to represent $N$ classical bits with a fingerprint of size $\Theta\left( \sqrt{N} \right)$ bits.

However, quantumly it was shown~\cite{buhrman2001quantum} that a fingerprint of size $n = O\left(\log N \right)$ qubits can encode a file of $N$ bits. This is an exponential advantage over the best possible classical fingerprint. This advantage has received extensive attention and further theoretical analysis \cite{yao2003power, golynski2005note, gavinsky2006strengths, de2004one, gavinsky2010quantum}. In addition, quantum fingerprinting has been demonstrated experimentally on NMR qubits \cite{du2006experimental}, as well as on photonic devices \cite{xu2015experimental} by adapting the original techniques to optical implementations.

While the potential of quantum fingerprinting is exciting, the fingerprint states described in \cite{buhrman2001quantum} are challenging to realize experimentally. The core contribution of Supercheq-EE is to demonstrate that random circuit sampling reproduces the asymptotic behavior of quantum fingerprinting. While our randomized approach is less efficient in absolute terms, it is (a) motivated by quantum experiments that can be run today and (b) scales to allow a graceful tradeoff between fingerprint size and preparation cost. Meanwhile, Supercheq-IE takes an alternative approach based on (hyper)graph states, which enables an \emph{incremental} fingerprinting protocol that (exceeds) matches the scaling of the best possible classical protocol.

\subsection{Distinguishing Fingerprints}
\label{subsec:fidelity_estimation}
Common to all quantum fingerprinting protocols, we need a protocol for distinguishing differing quantum fingerprints (or validating that candidate fingerprints are identical). The tools we need are furnished by fidelity estimation protocols, i.e.~for estimating the fidelity of two fingerprint states $\left\lvert\bra{\psi_i}\ket{\psi_j}\right\rvert^2$, which we review below.

\begin{figure}[t]
\centering
    \begin{quantikz}[row sep=0.3cm, column sep=0.45cm]
    \lstick{$\ket{0}$} & \gate{H} & \ctrl{2} & \gate{H} & \meter{} \\
    \lstick{$\ket{\psi_i}$} & \qw & \swap{1} & \qw & \qw \\
    \lstick{$\ket{\psi_j}$} & \qw & \swap{0} & \qw & \qw
\end{quantikz}
\caption{Swap-test circuit to estimate $|\braket{\psi_i}{\psi_j}|$.}
\label{fig:swap_test}
\end{figure}

The most common protocol for estimating the fidelity
is the standard SWAP test~\cite{schuld2018supervised}. Quantum circuit diagram of SWAP test is depicted in Fig.~\ref{fig:swap_test}. In this protocol, one initializes an ancilla qubit to $\ket{+}$, which can be done by applying a Hadamard gate on state $\ket{0}$ as illustrated in the figure. This ancilla is used as the control qubit of a controlled-SWAP that swaps corresponding qubits between the two fingerprint states $\ket{\psi_i},\ket{\psi_j}$. If global interaction gates are available on the target hardware, the controlled-SWAP can be implemented with a constant number of global interactions, independent of the size of the fingerprint registers \cite{gokhale2021quantum}. Upon applying a Hadamard gate to the ancilla qubit, the probability $P_0$ of measuring $0$ on the ancilla is then equal to:
\begin{equation}
    P_0=\frac{1+\left\lvert\bra{\psi_i}\ket{\psi_j}\right\rvert^2}{2}.
\end{equation}
If one is given multiple copies of the states $\ket{\psi_i}^{\otimes M},\ket{\psi_j}^{\otimes M}$, one can repeat the standard SWAP test $M$ times to test the equality of $\ket{\psi_i}$ and $\ket{\psi_j}$ with high probability. Note that if the two fingerprints are identical (because the files they represent are identical), then the standard SWAP test leaves the original copies of the states undisturbed. Therefore, if a referee concludes that Alice and Bob's fingerprints are identical, the fingerprint states can be returned to Alice and Bob and \textit{recycled} for future usage. 

In the $M$-copy setting there also exists a \emph{collective} SWAP test using a derangement operator~\cite{buhrman2001quantum}. This collective measurement can distinguish the two fingerprints with quadratically fewer copies than otherwise needed, which asymptotically saturates information-theoretic bounds.
More recent work has further optimized the collective measurement scheme, demonstrated its robustness to noise, and provided variants with relaxed experimental requirements~\cite{fanizza2020swap}.

In the photonic realm, variations of the Hong--Ou--Mandel effect for implementations with bosons~\cite{PhysRevLett.59.2044} provide another mechanism for estimating fidelities. Finally, the destructive SWAP test~\cite{PhysRevA.87.052330} reformulates the standard SWAP test in a more gate-efficient fashion that only requires gates between corresponding qubit pairs of $\ket{\psi_i}$ and $\ket{\psi_j}$, at the cost of the measurement being destructive. An additional advantage of the destructive SWAP test is that it only requires Clifford gates; as such, the graph variant of Supercheq-IE can be implemented entirely as a Clifford circuit, when coupled with the destructive SWAP test.

\section{Supercheq-EE (Efficient Encoding)} \label{sec:Supercheq-ee}

The core intuition behind Supercheq-EE (Efficient Encoding) is that Hilbert space is vast and can accommodate a massive number of \emph{approximately} distinguishable state vectors into a modest number of qubits. The original quantum fingerprinting \cite{buhrman2001quantum} protocol achieves exponential advantage over the best-possible classical protocol ($\log N$ vs.~$\sqrt{N}$ scaling) in a \textit{descriptive} fashion: it identifies a set of $N$ state vectors such that each pair of distinct state vectors has fidelity (squared-inner-product) less than a constant. Then, with qubit count that scales only logarithmically in inverse-error, fingerprints can be distinguished with small one-sided error. This can be achieved through one of many quantum circuits estimating the fidelity, as referenced in Sec.~\ref{subsec:fidelity_estimation}.

Our protocol builds on the result in \cite{buhrman2001quantum} in three ways. First, Supercheq-EE is \emph{prescriptive} rather than descriptive---the encodings for input states are explicitly given by specific choices of random quantum circuits. Second, as a corollary, Supercheq-EE is well-matched to recent experimental demonstrations of random circuit sampling such as quantum supremacy and quantum volume. We will demonstrate that there exists a procedure that constructs efficient fingerprints via circuits constructed out of local gates in 1D chosen uniformly at random. In contrast, generic preparations of arbitrary superpositions as in the descriptive procedure require large multiply-controlled operations that remain challenging on near-term hardware. Third, Supercheq-EE scales gracefully in circuit depth. In particular, if we restrict the circuit depth for preparation of quantum fingerprints to be some $\poly(n^\ell)$, we can still encode $n^\ell$-bit files. This feature stems from recent results pertaining to efficient sampling from approximate unitary $t$-designs~\cite{brandao2016local,harrow2018approximate,Haferkamp2022randomquantum}.

The generic approach of Supercheq-EE is summarized in Protocol~\ref{alg:Supercheq-ee}. Our results hinge on the fact that with high probability, large numbers of samples from certain distributions of random states have low (e.g. at most $0.5$) maximum pairwise fidelity with one another. We give specific prescriptions for generating these random states in Secs.~\ref{subsec:haar_random} and~\ref{subsec:approx_unitary_t_design}, along with a heuristic procedure in Sec.~\ref{subsec:hw_eff}. Once states from this distribution are generated, the referee can successfully discriminate fingerprints with high probability via a small (scaling logarithmically in target inverse-error) number of copies of each state. For simplicity, Protocol~\ref{alg:Supercheq-ee} specifically outlines the use of the standard SWAP test, but the other procedures in  Sec.~\ref{subsec:fidelity_estimation} may be used alternatively.

One technical point is that Alice and Bob must share the same random circuit sampling protocol in performing Protocol~\ref{alg:Supercheq-ee} such that identical files map to the same fingerprint. One approach for achieving this is for Alice and Bob to have a shared random key that is unknown to third parties; however, this setting would necessitate secure key exchange and regardless already has an efficient classical protocol~\cite{kremer1999randomized}. Instead, Alice and Bob can achieve this same result in principle by fixing the randomness of the protocol in advance so that at runtime, the protocol is deterministic. In effect, this turns the Supercheq-EE protocol into a randomly generated lookup table mapping the bits of each possible input file to a specific circuit. The lookup table can be generated with true randomness, such as with the outcomes of a quantum random number generator. In Protocol~\ref{alg:Supercheq-ee}, we describe this procedure as possessing \emph{a randomly generated dictionary} $U\left(\cdot\right)$ where, for each $\cdot$, either independently and identically distributed (i.i.d.) Haar-random or i.i.d. approximate $t$-design circuits are chosen. While this approach rigorously satisfies the constraints of the simultaneous message passing setting, the lookup table would be impractically large; we suggest a more practical approach involving pseudorandomness later in Sec.~\ref{subsec:hw_eff}, which we then validate numerically in Sec.~\ref{sec:simulation}.

\SetKwFunction{NumCopiesNeeded}{NumCopiesNeeded}
\begin{algorithm}
\kwSetup{Alice and Bob both possess $N$-bit files, $A$ and $B$. Alice and Bob both possess randomly generated dictionary $U\left(\cdot\right)$ producing circuits given files $\cdot$, as described in Sec.~\ref{sec:Supercheq-ee}.}
\kwOutput{Referee reports $A \stackrel{?}{=} B$ up to one-sided worst-case error $\epsilon$.}
$M \leftarrow$ \NumCopiesNeeded{$\epsilon$}.\\
$\ket{\psi_A}^{\otimes M} \leftarrow \left(U\left(A\right)\ket{0}^{\otimes n}\right)^{\otimes M}$ \tcp*[l]{Alice's FPs}
$\ket{\psi_B}^{\otimes M} \leftarrow \left(U\left(B\right)\ket{0}^{\otimes n}\right)^{\otimes M}$ \tcp*[l]{Bob's FPs}
Alice and Bob transmit fingerprints (FPs) to a referee.\\
Referee performs $M$ SWAP tests between pairs of FPs.\\
\uIf{all tests output $\ket{0}$ ancilla measurement}{
Conclude that $A$ = $B$.
Referee can return $\ket{\psi_A}^{\otimes M}$ and $\ket{\psi_B}^{\otimes M}$ to Alice and Bob to be recycled.
}
\Else{
Conclude that $A\neq B$.
}

\caption{Supercheq-EE}
\label{alg:Supercheq-ee}
\end{algorithm}

\subsection{Exponential Advantage by Haar-Random Sampling} \label{subsec:haar_random}

The exponential-advantage invocation of Supercheq-EE is capitulated in the following Theorem, proven in Appendix~\ref{app:proofs}.
\begin{restatable}{theorem}{Supercheqeethm}
Suppose we Haar-randomly draw $1.4^{2^n}$ $n$-qubit states. The probability that any pair has fidelity exceeding 0.5 vanishes as $n \to \infty$.
\label{thm:haar_random}
\end{restatable}

This result can be understood intuitively as a consequence of the birthday paradox---specifically because its square-root scaling is benign against an exponential advantage. In particular, we know through \cite{buhrman2001quantum} that it is possible to encode $2^{2^n}$ states~\footnote{Beware the double exponential; only one of the exponentials should be noteworthy because $n$ classical bits already naturally encode $2^n$ length-$n$ bitstrings.} (each corresponding to a length-$2^n$ bitstring) into $n$ qubits, such that the maximum fidelity between two states is smaller than a constant.

Intuitively, this means that the Hilbert space of $n$ qubits has $2^{2^n}$ approximately distinguishable containers. If we can draw from these containers in a Haar-random fashion, how many times can we draw before a collision? On the one hand, it would be extremely unlikely (though technically possible) that $2^{2^n}$ random draws would perfectly land in separate containers. However, the birthday paradox dictates that for $o\left( \sqrt{2^{2^n}} \right)$ draws, for example $1.4^{2^{n}}$, the collision probability approaches 0 in the limit of large $n$. This intuitive reasoning is formally proven in Appendix~\ref{app:proofs}. Though this randomized approach is less efficient than \cite{buhrman2001quantum} by a square root factor, it retains an exponential advantage in communication complexity over the optimal classical protocols.

Realizing an $n$-qubit Haar-random unitary (i.e.~from $\operatorname{SU}\left(2^n\right)$) requires $\sim\exp(n)$ cost in both the number of two-qubit gates and in classical random bits \cite{knill1995approximation}. This may be tolerable because Alice and Bob's files each comprise $\sim 2^n$ bits, so in this sense the fingerprint preparation takes time polynomial in the size of the file. However, it motivates the study of approximate Haar-random states, which can be achieved in much shallower circuit depth for any given $n$ and generally requires less entanglement.

\subsection{Polynomial Advantage by Approximate Unitary $t$-Design Sampling} \label{subsec:approx_unitary_t_design}
We now consider the preparation of quantum fingerprints via \emph{approximate unitary $t$-designs}. Informally, these are distributions over $\operatorname{SU}\left(2^n\right)$ whose first $t$ moments approximately equal to those of the Haar distribution. We give a more formal definition in Appendix~\ref{app:proofs}. These classes of random circuits are important theoretically as well as practically; previous proposals of unitary $t$-designs include creating secure quantum channels~\cite{hayden2004randomizing} and quantum physical unclonable functions~\cite{kumar2021efficient}. There exist known methods for efficiently sampling from approximate unitary $t$-designs via choosing local gates on a $d$-dimensional connectivity graph uniformly at random~\cite{brandao2016local,harrow2018approximate,Haferkamp2022randomquantum}. For simplicity, we here restrict to the 1D case.

Our key result is that this weaker notion of randomness suffices for generating fingerprints that are more memory-efficient than any classical fingerprint, proved in Appendix~\ref{app:proofs}.

\begin{restatable}{theorem}{Supercheqeetdesthm}
Let $\ell\geq 1$ be constant. There exists a class of 1D nearest-neighbor random circuits on $n$ qubits of depth $O\left(n^{5.01\ell-4.01}\right)$ such that by i.i.d. randomly drawing $1.4^{n^\ell}$ states prepared from this class, the probability that any pair has fidelity exceeding $0.5$ vanishes as $n\to\infty$.\label{thm:t_design}
\end{restatable}

Finally, we wish to comment on the cryptographic perspectives of Supercheq-EE with the help of approximate $t$-designs. A Haar-random unitary is computationally hard to prepare even for quantum devices. Since approximate $t$-designs are achievable by local random circuits, they are still distinct from Haar-random circuits from the computational complexity perspectives. In fact, one could define \emph{psedorandom quantum states} \cite{ji2018pseudorandom}, which are quantum states that cannot be distinguished from Haar randomness within polynomial operations. Fingerprinting made from pseudorandom quantum states might be significantly robust from potential hacks. However, there are deep connections between approximate $t$-designs and pseudorandom quantum states that are still developing in progress (see \cite{ananth2022cryptography,bouland2022quantum,laracuente2026approximate,doi:10.1126/science.adv8590}). Moreover, there are also important problems about whether local random circuits will scramble linearly towards approximate $t$-designs. The property of linear scrambling of local random circuits is a deep question relating to black hole physics, complexity and cryptography \cite{hayden2007black,bouland2019computational,brandao2021models}. It is claimed that local random circuits will both scramble linearly in $t$ and the number of qubits, which, will at the same time, provide a linear complexity growth \cite{hunter2019unitary,brandao2021models,haferkamp2022linear,liu2022estimating,ma2024construct,metger2024simple,doi:10.1126/science.adv8590}. We leave the questions between pseudorandom quantum states, complexity growth, and our quantum fingerprinting for future research.

\subsection{Empirical Advantage by Hardware-Efficient Sampling} \label{subsec:hw_eff}

We now discuss relaxations of the protocol described in Protocol~\ref{alg:Supercheq-ee} to aid in practical implementations of Supercheq-EE. We will show in simulations in Sec.~\ref{sec:simulation} that in practice, these relaxations achieve low overlap between distinct quantum fingerprints, as desired.

First, we consider shallower classes of quantum random circuits than those required in Theorem~\ref{thm:t_design} for generating circuits from an approximate $t$-design. The proof of Theorem~\ref{thm:t_design} relies on the best-known convergence results for sampling from such a distribution in 1D, which gives a depth of $O\left(t^{4.01}\left(nt+\log\left(\epsilon^{-1}\right)\right)\right)$ for sufficiently large $n$ when $t=O\left(\poly\left(n\right)\right)$\footnote{Unfortunately, as we require extremely tight errors---i.e., $\epsilon=\exp\left(-\operatorname{poly}\left(n\right)\right)$---we are outside of the regime of more efficient recent developments in implementing $t$-designs~\cite{10756150,doi:10.1126/science.adv8590,schuster2025strongrandomunitariesfast}.}~\cite{Haferkamp2022randomquantum}. This is achieved via choosing uniformly at random nearest-neighbor gates in 1D in a ``brickwork'' pattern, i.e.~alternating between gates on even pairs and odd pairs. However, it is conjectured that such circuits scramble to $\epsilon$-approximate $t$-designs in depth $O\left(nt+\log\left(\epsilon^{-1}\right)\right)$~\cite{brandao2016local, hunter2019unitary}. Assuming this scaling, the circuit depth polynomial in Theorem~\ref{thm:t_design} reduces to $O\left(n^\ell\right)$. Furthermore, the upper-bounds on the fidelity in Theorems~\ref{thm:haar_random} and~\ref{thm:t_design} are proven using standard worst-case bounds that may potentially be loose; if so, the required circuit depths needed in practice may be even further reduced.

Second, though the results of \cite{Haferkamp2022randomquantum} require specific classes of random quantum circuits in 1D, we conjecture that any ``sufficiently random'' local quantum circuit yields sufficient scrambling for low fidelity between different fingerprints. We refer to such circuit families as the Hardware-Efficient variant of Supercheq, in which a practitioner would pick a scrambling circuit based on the the available connectivity and gateset of target hardware. We investigate this approach with a concrete example in Sec.~\ref{sec:simulation}, where we consider a class of random quantum circuits suited for typical neutral atom (cold atom) systems such as Infleqtion's Sqale neutral atom quantum computer \cite{radnaev2025universal, bedalov2024fault, rines2025demonstration}. In future work, we suggest considering global interaction gatesets available on atomic systems, such as the Rydberg blockade which can be used to perform entangling gates spanning 3+ qubits simultaneously \cite{isenhower2011multibit}.

Finally, we relax the cryptographic assumption of true randomness (e.g. by quantum random number generation) that was used to assure the fingerprinting performance and security of Theorems~\ref{thm:haar_random} and~\ref{thm:t_design}. This assumption would be costly in practice, because it would require an inefficient lookup-table description of the fingerprinting protocol for any instance of randomness. Instead, we will use the file to be fingerprinted as a seed for pseudorandomly constructing quantum fingerprinting circuits. Empirical tests that validate the performance of this relaxation are presented in Sec.~\ref{sec:simulation}.

\section{Supercheq-IE (Incremental Encoding)} \label{sec:Supercheq_ie}

One limitation of previous fingerprinting protocols (including Supercheq-EE) is the inability to perform \emph{incremental updates} on fingerprints. In particular, if a fingerprint is generated for a file, and the file is subsequently modified by even a single bit flip, typically the fingerprint for the updated file must be generated from scratch.

With this motivation, we now describe Supercheq-IE, a protocol for generating quantum fingerprints that can be incrementally updated. Specifically, if file $A$ is modified by incremental (constant-space) changes to yield file $A'$, then the corresponding fingerprint $\ket{\psi_A}$ of $A$ can be modified with constant cost to obtain the fingerprint $\ket{\psi_{A'}}$ of file $A'$. There are two subvariants to Supercheq-IE: graph and hypergraph. For simplicity, we begin by focusing on the graph variant where Supercheq-IE achieves a communication complexity of $\Theta\left(\sqrt{N}\right)$ qubits in the SMP setting, matching the optimal classical protocol; however, to the best of our knowledge no incremental classical protocol achieves this bound without cryptographic assumptions. For example, previous incremental approaches \cite{bellare1994incremental} have relied on the hardness of discrete logarithm---which in fact has an efficient quantum algorithm \cite{shor1999polynomial}. Subsequent work \cite{phan2006security} has noted other security issues with attempting to construct incremental hash functions; Supercheq-IE averts these issues by relying on an exact one-to-one mapping from files to quantum fingerprints.

\begin{algorithm}
\SetKw{KwTo}{to}
\SetKwFunction{AddEdge}{AddEdge}
\kwSetup{Alice and Bob both possess $N$-bit files, A and B.}
\kwOutput{Referee reports $A \stackrel{?}{=} B$ up to one-sided worst-case error $\epsilon$.}
$n\gets \lceil (\sqrt{8N + 1} + 1)/2 \rceil$ \\
$G_\text{Alice} \gets \text{n-vertex graph}$ \tcp*[l]{(repd by nxn adj mat)}
\For{$i \leftarrow 0$ \KwTo $N$}{
\For{$row \leftarrow 1$ \KwTo $n$}{
\For{$col \leftarrow 0$ \KwTo $row$}{
\If{A[i]}{\AddEdge{$G_\text{Alice}$, row, col};\\}
i++;
}}}
Alice produces $n$-qubit graph state $\ket{\psi_A}$ corresponding to $G_\text{Alice}$: she initializes each qubit to $\ket{+}$ and then applies $\operatorname{CZ}(i, j)$ for every edge $(i, j)$.\\
Alice repeats this $M$ times to produce $\ket{\psi_A}^{\otimes M}$.\\
Bob similarly performs the above to produce $\ket{\psi_B}^{\otimes M}$\\
Alice and Bob transmit fingerprints to a referee.\\
Referee performs $M$ SWAP tests between pairs of FPs.\\
\uIf{all tests output $\ket{0}$ ancilla measurement}{
Conclude that $A$ = $B$.
Referee can return $\ket{\psi_A}^{\otimes M}$ and $\ket{\psi_B}^{\otimes M}$ to Alice and Bob to be recycled.
}
\Else{
Conclude that $A\neq B$.
}

\caption{Supercheq-IE (graph subvariant)}
\label{alg:Supercheq_ie}
\end{algorithm}

The essential ingredients for the graph subvariant of Supercheq-IE are twofold. First, we demonstrate that any $N$-bit file can be encoded into an $n$-qubit \textit{graph state} (defined later), where $n$ scales as $\sim \sqrt{2N}$. We refer to this as the file-to-graph encoding, described in full below. Second, we observe that distinct qubit graph states have bounded fidelity. Namely, for two distinct graph states $\ket{\psi_A},\ket{\psi_B}$, $\left\lvert\bra{\psi_A}\ket{\psi_B}\right\rvert^2\leq\frac{1}{2}$~\cite{PhysRevA.70.052328}, allowing efficient equality testing (with low one-sided error) via one of the SWAP tests described in Sec.~\ref{subsec:fidelity_estimation}.

Taken together, these two ingredients lead to a natural protocol for quantum fingerprinting. Alice and Bob encode their $N$-bit files $A$ and $B$ into $M \in O\left(\log\left(\epsilon^{-1}\right)\right)$ copies of graph states, each with $O\left(\sqrt{N}\right)$ qubits. Thus, with $O\left(\sqrt{N}\log\left(\epsilon^{-1}\right)\right)$ qubits of communication, a referee is able to determine whether $A$ and $B$ are equal with probability at least $1-\epsilon$ via $O\left(\log\left(\epsilon^{-1}\right)\right)$ standard SWAP tests. The full procedure for the graph subvariant of Supercheq-IE is summarized in Protocol~\ref{alg:Supercheq_ie}.

The graph subvariant should be viewed as the simplest instance of a broader incremental-encoding principle. In the next subsection, we generalize the construction from ordinary graph states to hypergraph states, where each file bit is encoded by the presence or absence of a higher-order hyperedge rather than a pairwise edge. This extension preserves the key incremental-update property, since a single bit flip is still implemented by a constant-size controlled-phase operation, while allowing the fingerprint size to scale as $O(N^{1/\ell})$ for $\ell$-uniform hypergraphs. This yields arbitrarily large polynomial improvements in communication for fixed $\ell$.

\subsection{File-to-Graph Encoding}

The file-to-graph encoding, depicted in  Fig.~\ref{fig:incremental_encoding_overview}, begins with the observation that any $N$-bit file can be encoded in row-major order as the entries of a strictly lower triangular $n\times n$ matrix. The matrix dimension $n$ is chosen such that the number of available positions in the lower triangle is sufficient to store all $N$ bits yielding, $n=\left\lceil\frac{1}{2}\left(\sqrt{8N+1}+1\right)\right\rceil$ which scales as $\sim \sqrt{2N}$. Adding this matrix to its transpose produces a symmetric matrix that can then be interpreted as the adjacency matrix of a graph $G$, with $n$ vertices. The graph $G$ can then be associated with a qubit \emph{graph state} on $n$ qubits~\cite{PhysRevA.68.022312}. The graph's adjacency matrix $G_{ij}$ defines a state stabilized by the Pauli strings
\begin{equation}
    S_i=X_i\prod\limits_{j\neq i} Z_j^{G_{ij}}
\end{equation}
for $i=1,\ldots,n$. Next, we summarize the full circuit implementation of the Supercheq-IE graph subvariant, including the construction of the corresponding graph state.

\begin{figure}
    \centering
    \includegraphics[width=0.4\textwidth]{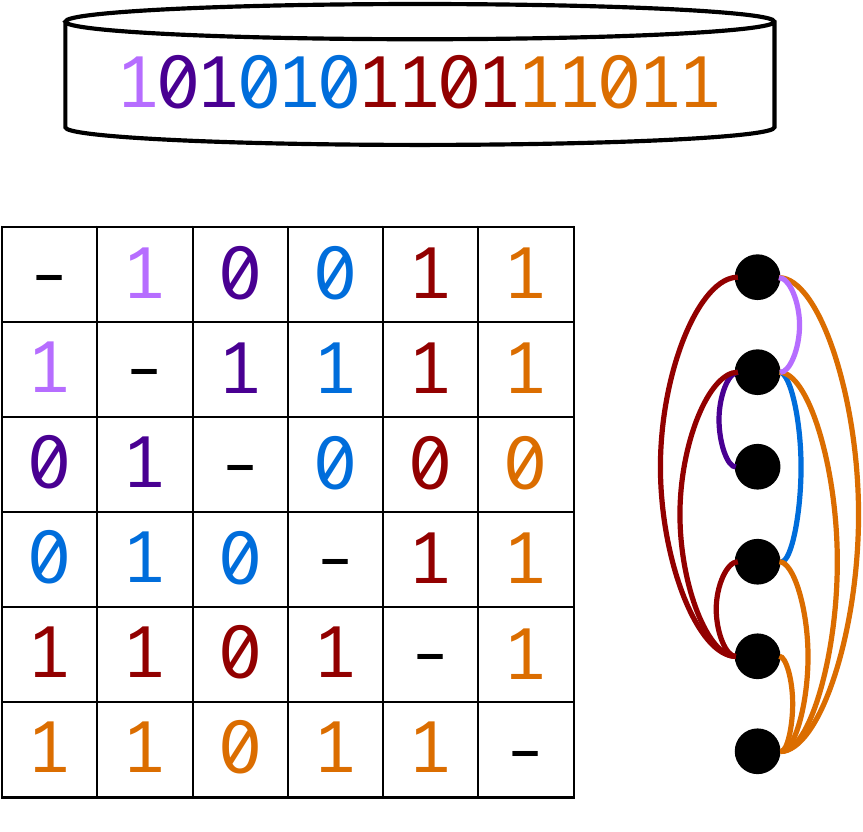}
    \caption{Overview of Supercheq-IE's file-to-graph encoding. In this example, an $N=15$-bit file is encoded into a graph with $n=6$ vertices (asymptotically, $\sim\sqrt{2N}$). The corresponding 6-qubit graph state is the Supercheq-IE fingerprint.}
    \label{fig:incremental_encoding_overview}
\end{figure}

\begin{figure}
    \centering
    \includegraphics[width=0.5\textwidth]{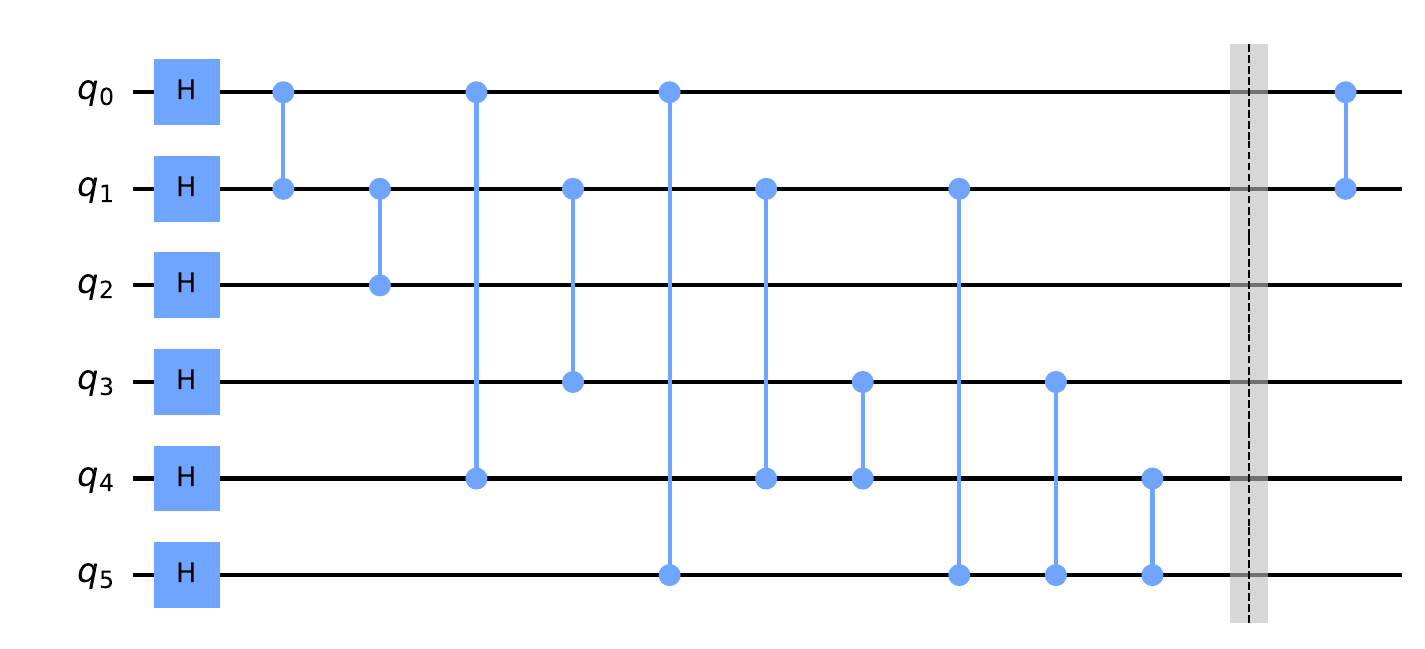}
    \caption{Quantum circuit for graph subvariant of Supercheq-IE, demonstrating an incremental bit flip to the original file \texttt{101010110111011} from Fig.~\ref{fig:incremental_encoding_overview}. Each of the \texttt{1}'s in the original file corresponds to a $\operatorname{CZ}$ gate between two qubits. The last $\operatorname{CZ}$, after the barrier, models the incremental update that would be necessary if the first bit of the file was flipped from from \texttt{1} to \texttt{0}.}
    \label{fig:Supercheq_ie_circuit}
\end{figure}

\subsection{Quantum Circuit Implementation}
Fig.~\ref{fig:Supercheq_ie_circuit} depicts a full example of Supercheq-IE's quantum circuit encoding. As also referenced in Protocol~\ref{alg:Supercheq_ie}, a graph state can be created by applying the Hadamard gate to each qubit, followed by a $\operatorname{CZ}$ corresponding to any row-column pair of the adjacency matrix that is 1 (i.e. corresponding to any edge in the graph). Subsequently, any bit flip to the source file can be performed with a single $\operatorname{CZ}$ to the appropriate pair of qubits. Because $\operatorname{CZ}$ gates commute with each other, the updated graph state faithfully reflects the updated file.

Finally, when requested, Alice and Bob's graph state fingerprints can be sent to a referee. As in Supercheq-EE, the referee can use any of the protocols in Sec.~\ref{subsec:fidelity_estimation} to check whether Alice and Bob's files are identical. If so, the fingerprints (graph states) can be recycled and returned to Alice and Bob. Recycling is particularly well-matched to the incremental scenario: if any updates were made to Alice's source file while the referee was examining fingerprints, she can apply those updates when she receives the recycled fingerprint instead of needing to generate it from scratch.

Note that supporting bit flips is also sufficient to support general writes, which can be implemented as a read on the source file followed by a conditional flip. For instance, setting some bit to 0 is equivalent to reading it and flipping it iff the bit is currently 1. In addition, Supercheq-IE can support scenarios where the source is resized to a larger file, simply by adding a new qubit. This is equivalent to adding a new vertex to the graph state and therefore a new row to the adjacency matrix. We also note the possibility of supporting yet more exotic incremental updates. For example, one can use two ancilla qubits to complement all edges between two vertex subsets in time scaling only linearly in the number of qubits in the sets, rather than in the number of edges; said otherwise, two ancilla qubits can be used to complement a block of the graph's adjacency matrix in a time that is linear in the dimensions of the block, rather than its volume~\cite{PhysRevA.93.032314}. In summary, between individual bit flips, writes (by read-and-conditional-flip), resizing, and even more exotic operations, Supercheq-IE can support incremental updates for quite general changes to source files.

In a real-world scenario, we envision that Alice and Bob could maintain graph states fingerprints corresponding to their databases. As updates are made to their corresponding databases, Alice and Bob can perform constant-cost incremental updates to their graph states, possibly with modest batching for efficiency or to leverage available parallelism while executing multiple $\operatorname{CZ}$ gates.

As noted earlier, the encoding circuit of the graph subvariant of Supercheq-IE, as well as the destructive SWAP test for distinguishing fingerprints, are both Clifford circuits. In the context of fault-tolerant implementations, this is significant because Clifford gates---which do not need costly magic state distillation---are much cheaper to implement than non-Clifford gates. For instance, a Toffoli (non-Clifford gate) is roughly 100x slower than a Clifford gate in the surface code by some estimates \cite{babbush2021focus}. As such, we propose that Supercheq-IE should be one of the first target applications for early fault-tolerant devices.

\subsection{Hypergraph Encoding}
The quadratic information-compression factor in Supercheq-IE's graph subvariant can be boosted to an arbitarily-high polynomial advantage by \textit{hypergraphs}, which generalize traditional graphs with \textit{hyperedges} that can contain more than two vertices. We use the notation $\ell$-hypergraph to denote the case where each hyperedge has cardinality $\ell$.

\begin{restatable}{theorem}{Supercheqiethm}
Let $x\in\{0,1\}^N$ be a classical file with $N$ bits. For any integer $\ell \geq 2$, there exists an encoding of $x$ as an $\ell$-hypergraph quantum state on $O(N^{1/\ell})$ qubits that supports constant-cost updates when individual bits of $x$ change. Compared to classical fingerprinting methods with best-case scaling of $O(\sqrt{N})$ \cite{kremer1999randomized, newman1996public}, Supercheq-IE's hypergraph encoding achieves polynomial improvement, $O(\sqrt[\ell]{N})$.
\label{thm:hypergraph-encoding}
\end{restatable}

The proof sketch is as follows. For a $\ell$-hypergraph with $n$ vertices, the number of possible hyperedges is $n$ choose $\ell$, which scales as $O(n^\ell)$. In the hypergraph, vertices correspond to qubits, and an $\ell$-hyperedge phase interaction encodes one bit of $x$. The $N$-bit string is partitioned onto indices corresponding to the hyperedges of the $\ell$-uniform hypergraph. Since the number of  $\ell$-hyperedges on $n$ vertices scales $O(n^\ell)$, choosing $n= N^{1/\ell}$ provides sufficient capacity to encode $N$ bits. Further, updates to individual bits are implemented by applying a $\ell$-qubit controlled phase operation, so the encoding supports constant-cost updates when $x$ changes. Thus, the encoding requires $O(\sqrt[\ell]{N})$ while preserving constant-cost incremental updates.

To make this construction precise, we now describe the file-to-hypergraph encoding used in Supercheq-IE. The key concepts for applying a file-to-hypergraph encoding in Supercheq-IE are the same: we associate each bit of a file with the presence or absence of a hyperedge $e$ containing more than two vertices. Rather than being represented by an adjacency matrix, a hypergraph is specified by its set of hyperedges $E$, where each hyperedge $e \in E$ is an unordered subset of vertices. For example, the hyperedge $e=\{i,j,k\}$ is invariant under permutations of its vertices, so $(i,j,k)$, $(k,j,i)$, and all other permutations correspond to the same hyperedge.

Then, given an $\ell$-hypergraph on $n$ vertices, we associate a hypergraph state on $n$ qubits that naturally generalizes the traditional graph state. In particular, we initialize the qubits in the state $\ket{+}^{\otimes n}$, which is obtained by applying Hadamard gates to qubits initialized in $\ket{0}^{\otimes n}$. Then, to create the hypergraph state $\ket{\psi_F}$ associated with a file $F$, we apply $\operatorname{C^{\ell-1}Z}(e)$ for each hyperedge $e \in E_F$. The operator $\operatorname{C^{\ell-1}Z}(e)$ is a multi-controlled $Z$ gate acting on the qubits associated with the vertices in $e$, with $\ell-1$ control qubits. Formally, we define the operator $G_F$ as the generator associated with the file $F$:
\begin{equation}
    \mathcal{G}_F = \prod_{e \in E_F} \operatorname{C}^{\ell-1}Z(e),
\end{equation}
so that
\begin{equation}
    \ket{\psi_F} = \mathcal{G}_F \ket{+}^{\otimes n}.
\end{equation}

Notice that for $\ell=2$, this recovers the traditional Supercheq-IE graph subvariant, where each edge $(i, j)$ has a corresponding $\operatorname{CZ}$ gate between qubits $i$ and $j$. For larger $\ell$, the necessary $\operatorname{C^{\ell-1}Z}$ gate comes from the $\ell-1$ level of the Clifford hierarchy. Using quantum error correction codes such as the hypercube code \cite{kubica2015unfolding, hangleiter2025fault}, these gates can be implemented cheaply via transversal physical operations.

Finally, we use a SWAP test to distinguish between two hypergraph states, $\ket{\psi_A}$ and $\ket{\psi_B}$, associated with files $A$ and $B$. As discussed before, the SWAP test is used to infer the squared overlap, $\left|\braket{\psi_A}{\psi_B}\right|^2$. The overlap between the two states is
\begin{equation}
    \braket{\psi_A}{\psi_B} = \bra{+}^{\otimes n} \mathcal{G}_A^\dagger \mathcal{G}_B \ket{+}^{\otimes n}.
\end{equation}
If the files $A$ and $B$ are identical, then $\mathcal{G}_B = \mathcal{G}_A$, and it follows that the state overlap is $1$. If the files are different, their overlap is less than $1$. The maximum overlap between two distinct states associated with two different files occurs when the files differ by only one bit. In that case, $\mathcal{G}_A$ and $\mathcal{G}_B$ differ by only a single multi-controlled $Z$ gate. That is,
\begin{equation}
    \mathcal{G}_B = \mathcal{G}_A \operatorname{C^{\ell-1}Z}(e_d),
\end{equation}
where $e_d$ is the hyperedge associated with the differing bit. In that case, we have
\begin{align}
    \braket{\psi_A}{\psi_B}
    &= \bra{+}^{\otimes n} \mathcal{G}_A^\dagger \mathcal{G}_A \operatorname{C^{\ell-1}Z}(e_d) \ket{+}^{\otimes n}, \notag \\
    &= \bra{+}^{\otimes n} \operatorname{C^{\ell-1}Z}(e_d) \ket{+}^{\otimes n},
\end{align}
which shows that the overlap is the expectation value of $\operatorname{C^{\ell-1}Z}(e_d)$ in the uniform superposition state $\ket{+}^{\otimes n}$. This expectation value is just the average of the diagonal entries of $\operatorname{C^{\ell-1}Z}$ on the $\ell$ qubits in $e_d$. Since this gate has eigenvalue $-1$ only on $\ket{1}^{\otimes \ell}$ and eigenvalue $+1$ on the remaining $2^\ell-1$ computational basis states, the average is $1- 1 /2^{\ell-1}$. For example, the maximum possible overlap is $75\%$ for distinct $3$-hypergraphs and $87.5\%$ for distinct $4$-hypergraphs. More generally, for any fixed $\ell$, distinct hypergraph states remain separated by a constant gap, which is enough to distinguish them with one-sided error exponentially small in the number of state copies.
For example, the maximum possible overlap is $75\%$ for distinct $3$-hypergraphs and $87.5\%$ for distinct $4$-hypergraphs. More generally, for any fixed $\ell$, distinct hypergraph states remain separated by a constant gap, which is enough to distinguish them with one-sided error exponentially small in the number of state copies.

To demonstrate a practical invocation of our hypergraph formulation of Supercheq-IE, we illustrate the encoding for a 20-bit file in Figure~\ref{fig:hypergraph_supercheq}. Since ${6 \choose 3} = 20$, we can represent the file with a 3-hypergraph on 6 vertices, and in turn, a 6-qubit state. Observe that a traditional graph on 6 vertices could only represent a ${6 \choose 2} = 15$ bit file. Moreover, this advantage scales aggressively for higher $N$. For instance, a 1 MB ($N = 8 \times 2^{30}$) file would require a 133k qubits with a graph state, but only a 50 qubits with a 10-hypergraph state. For completeness, Protocol~\ref{alg:supercheq_ie_hypergraph} explicitly describes the full procedure for Supercheq-IE with hypergraph encoding.

\begin{figure}
\centering
\includegraphics[width=0.5\textwidth]{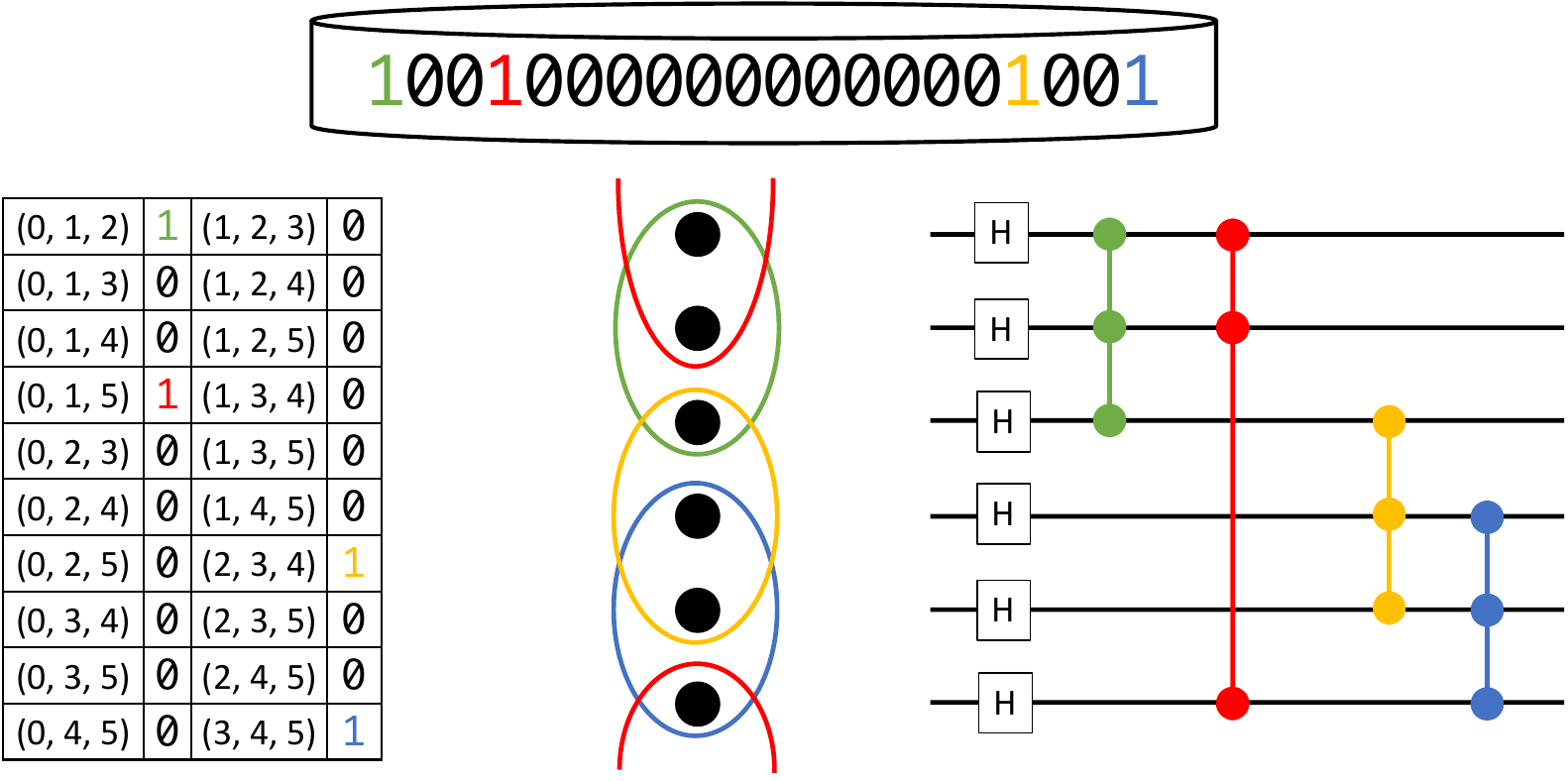}
\caption{Supercheq-IE hypergraph encoding. In this example, a 20-bit file is encoded into the ${6 \choose 3} = 20$ possible hyperedges of a 3-hypergraph on 6 vertexes. The corresponding hypergraph state preparation circuit on 6 qubits has a $\operatorname{CCZ}$ gate (per convention, depicted as three controls) for each of the four hyperedges.}
\label{fig:hypergraph_supercheq}
\end{figure}

\begin{algorithm}
\SetKw{KwTo}{to}
\SetKwFunction{AddHyperedge}{AddHyperedge}
\kwSetup{Alice and Bob both possess $N$-bit files, A and B. Agreed hyperedge cardinality is $\ell$.}
\kwOutput{Referee reports $(A \stackrel{?}{=} B)$ up to one-sided worst-case error $\epsilon$.}
$n\gets \text{smallest }n\in\mathbb{N}\text{ such that }\binom{n}{\ell}\geq N$ \\
$G_\text{Alice} \gets \text{empty $n$-vertex graph}$  \\
\For{$i \leftarrow 0$ \KwTo $n \choose \ell$}{
\If{A[i]}{\AddHyperedge{$G_\text{Alice}$, i\textsuperscript{th} sorted hyperedge};\\}
i++;
}
Alice produces $n$-qubit hypergraph state $\ket{\psi_A}$ corresponding to $G_\text{Alice}$: she initializes each qubit to $\ket{+}$ and applies $\operatorname{C^{\ell-1}Z}(e_1, ..., e_{\ell})$ for every edge $e$.\\
Alice repeats this $M$ times to produce $\ket{\psi_A}^{\otimes M}$.\\
Bob similarly performs the above to produce $\ket{\psi_B}^{\otimes M}$\\
Alice and Bob transmit fingerprints to a referee.\\
Referee performs $M$ SWAP tests between pairs of FPs.\\
\uIf{all tests output $\ket{0}$ ancilla measurement}{
Conclude that $A$ = $B$.
Referee can return $\ket{\psi_A}^{\otimes M}$ and $\ket{\psi_B}^{\otimes M}$ to Alice and Bob to be recycled.
}
\Else{
Conclude that $A\neq B$.
}

\caption{Supercheq-IE (hypergraph subvariant)}
\label{alg:supercheq_ie_hypergraph}
\end{algorithm}

\section{Simulation (Supercheq-EE)} \label{sec:simulation}

In Sec.~\ref{subsec:hw_eff}, we give heuristic reasons that shallow-depth local random quantum circuits beyond those considered in Theorem~\ref{thm:t_design} can achieve competitive fingerprinting efficiency. We also suggest using seeded pseudorandomness to obviate an impractically large, randomly generated lookup table for constructing fingerprints. Here, we validate these assumptions on a variety of random quantum circuits over multiple file sizes.

Supercheq-EE's high (up to exponential) encoding efficiency poses a computational challenge for simulations testing its performance. For example, if we have $n$ qubits, one could fingerprint roughly $1.4^{2^n}$ distinct files. Verifying that the inner product between pairs is small would incur a quadratic overhead on top of this double exponential, in the sense that we would need to compute roughly $(1.4^{2^n} \times 1.4^{2^n})/2$ inner products. While this scaling appears fundamentally unavoidable---owing to the high efficiency of Supercheq-EE---it motivates us to consider GPU-accelerated approaches.

\subsection{GPU-Accelerated Simulation}\label{subsec:noiseless_sim}
We leveraged GPU simulation for two core functions: state vector simulation and fidelity computation. Our simulations were performed with an NVIDIA A100 GPU.

\begin{figure*}
    \centering
    \subfloat[\label{subfig:schematic}]{%
        \includegraphics[width=0.85\textwidth]{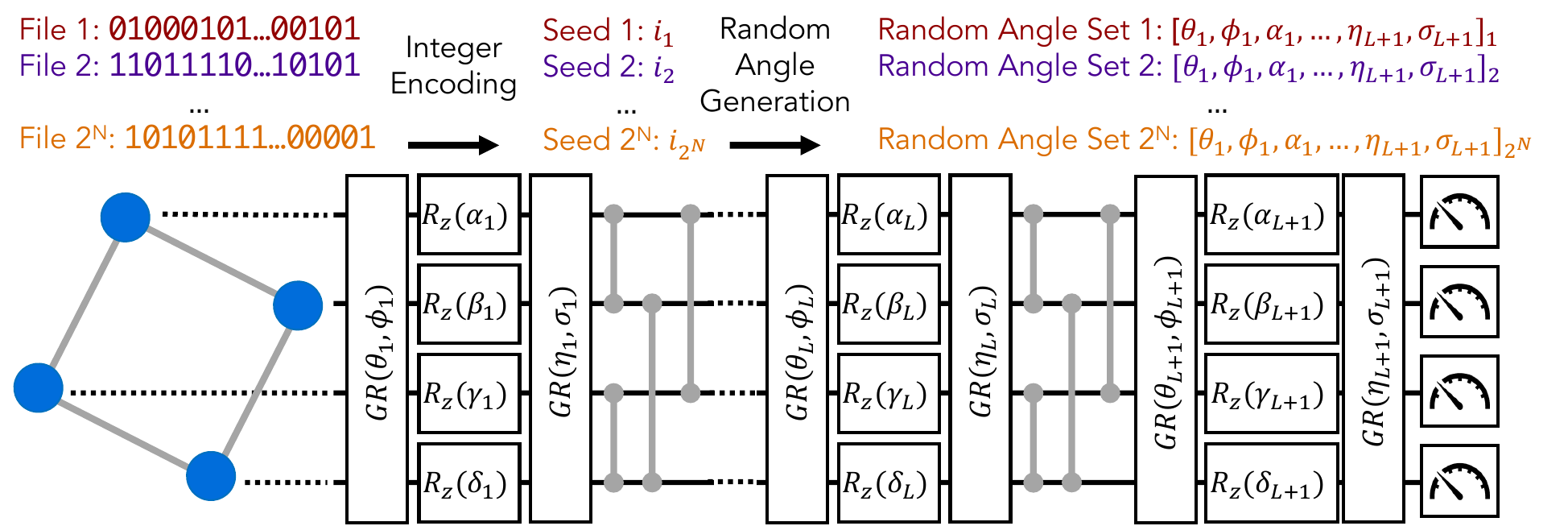}%
    }\hfill
    \subfloat[\label{subfig:nearest-neighbor}]{%
        \includegraphics[width=0.5\textwidth]{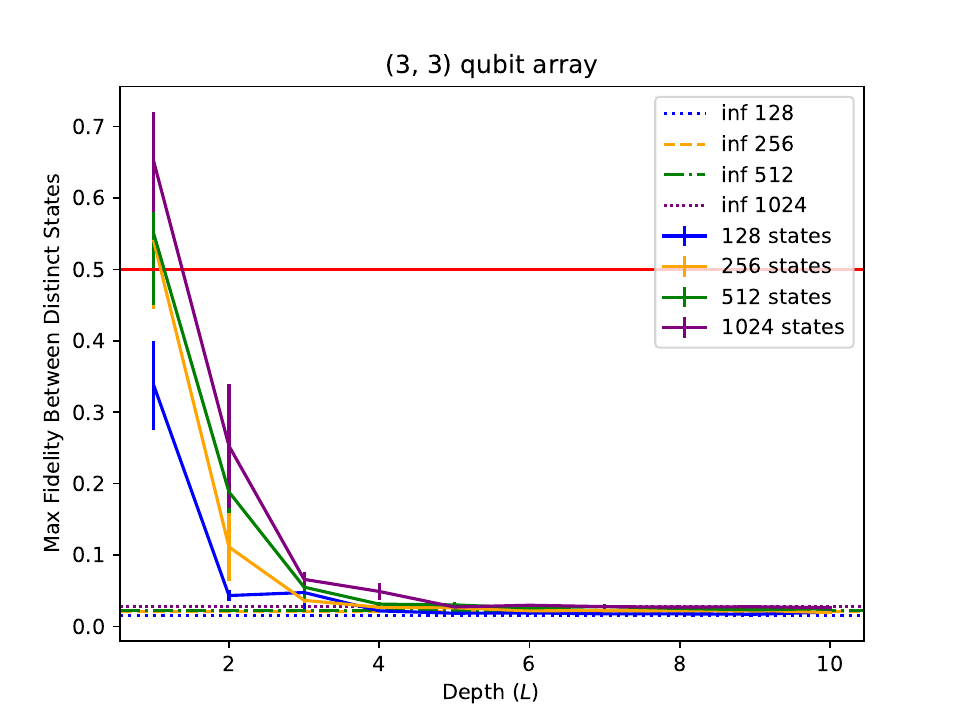}%
    }
    \subfloat[\label{subfig:fully-connected}]{%
        \includegraphics[width=0.5\textwidth]{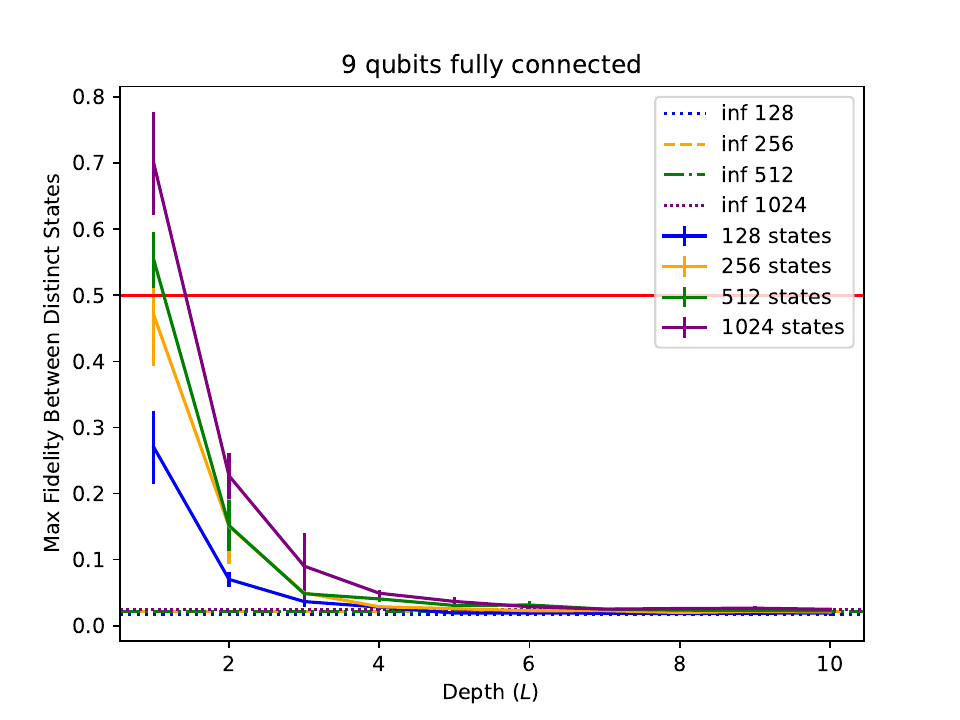}%
    }\hfill
    \caption{GPU-based simulations of Supercheq-EE. (a) Pseudorandom circuit models simulated. (b,c) The max fidelity (lower is better) between pairs of distinct fingerprints generated by the pseudorandom quantum circuits described in Sec.~\ref{subsec:noiseless_sim} is given for (b) 2D nearest-neighbor connectivity and (c) all-to-all connectivity. Error bars denote standard deviations over five trials. Dashed lines demonstrate the average performance of true Haar-random quantum states over five instances. The red line denotes an overlap of $\frac{1}{2}$ as a guide to the eye that matches the maximum overlap between distinct fingerprints in Supercheq-IE.
    \label{fig:noiseless_sims}}
\end{figure*}

First, to compute state vectors for fingerprints, we leverage the cuStateVec library in the cuQuantum SDK \cite{stanwyck2022cuquantum}. We first consider using shallow-depth local random quantum circuits to generate quantum fingerprints in the noiseless regime. In the language of Sec.~\ref{subsec:hw_eff}, we consider a hardware-efficient encoding circuit inspired by neutral atom systems. This encoding circuit is given by the layered application of a uniformly random global rotation gate:
\begin{equation}
\operatorname{GR}\left(\theta,\phi\right)=\exp\left(-\frac{i}{2}\sum\limits_{j=1}^n\left(\cos\left(\phi\right)X_j+\sin\left(\phi\right)Y_j\right)\right),
\end{equation}
uniformly random single-qubit $Z$ rotations, another uniformly random global rotation gate, and pairs of $\operatorname{CZ}$ gates according to various qubit connectivities. We then consider a final application of the random single-qubit gates. We pseudorandomly choose each rotation angle via seeding the NumPy~\cite{harris2020array} pseudorandom number generator with the $N$-bit file to be fingerprinted. We give an overview of this practical implementation of Supercheq-EE in Fig.~\ref{fig:noiseless_sims}(a). After evolving under the random quantum circuit, we store cuStateVec output fingerprints into the rows of a large matrix $\bm{X}$.

Next, we note that the task of computing fidelities between pairs of fingerprints reduces to the task of computing the matrix product of overlaps
\begin{equation}
    \bm{O}=\bm{X}\bm{X}^\dagger.
\end{equation}
The fidelity matrix $\bm{F}$ follows directly from $\bm{O}$ by taking the elementwise squared-magnitudes.

In Fig.~\ref{fig:noiseless_sims}(b,c), we plot the the max fidelity between all pairs of 9-qubit quantum fingerprints generated by this procedure for two different connectivities: nearest-neighbor on a 3-by-3 grid and fully connected. For each circuit depth, we perform five trials seeded with different nonces appended to the file bits. We also plot the performance of Haar-random sampling seeded by each of the files as a benchmark.

As anticipated in Sec.~\ref{subsec:hw_eff}, we see that the random circuit performance matches that of the Haar-random sampling even at relatively shallow depths and that the practical pseudorandom seeding achieves low overlap between distinct fingerprints. For instance, at a depth of just $5$ layers, all trials in the nearest-neighbor case with 1024 input states have a maximum pairwise fingerprint fidelity below 0.05 (across 524k pairs). By a depth of just 7 layers, the performance is visually indistinguishable from true Haar-random sampling. Moreover, we note that the performance between the nearest-neighbor and fully-connected cases is nearly indistinguishable. Although this may be an artifact of the small diameter of the 3-by-3 layout, it is nonetheless encouraging for hardware-efficient approaches.

Generating $\bm{X}$ matrices for the simulations in Fig.~\ref{fig:noiseless_sims}(b,c) required billions floating point operations, which GPU computation can easily handle in seconds. In ongoing work, we are running even larger simulations with closer to a million states (files) encoded into 9-qubit fingerprints with Supercheq-EE.

\subsection{Noisy Simulation}\label{subsec:noisy_sim}

\begin{figure}[t]
    \centering
    \subfloat[\label{subfig:noise-ideal}]{%
        \includegraphics[width=0.5\linewidth]{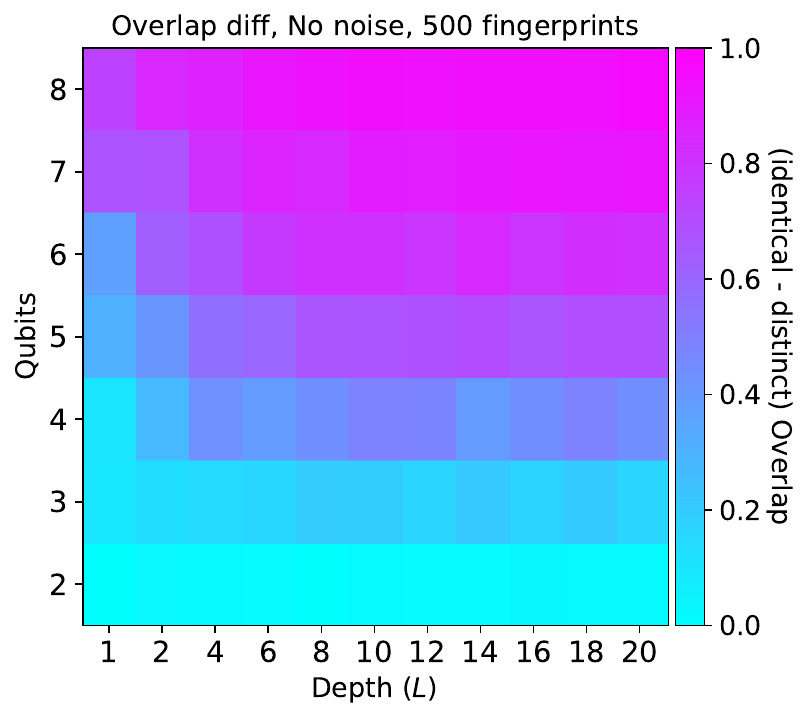}%
    }
    \subfloat[\label{subfig:noise-coherent}]{%
        \includegraphics[width=0.5\linewidth]{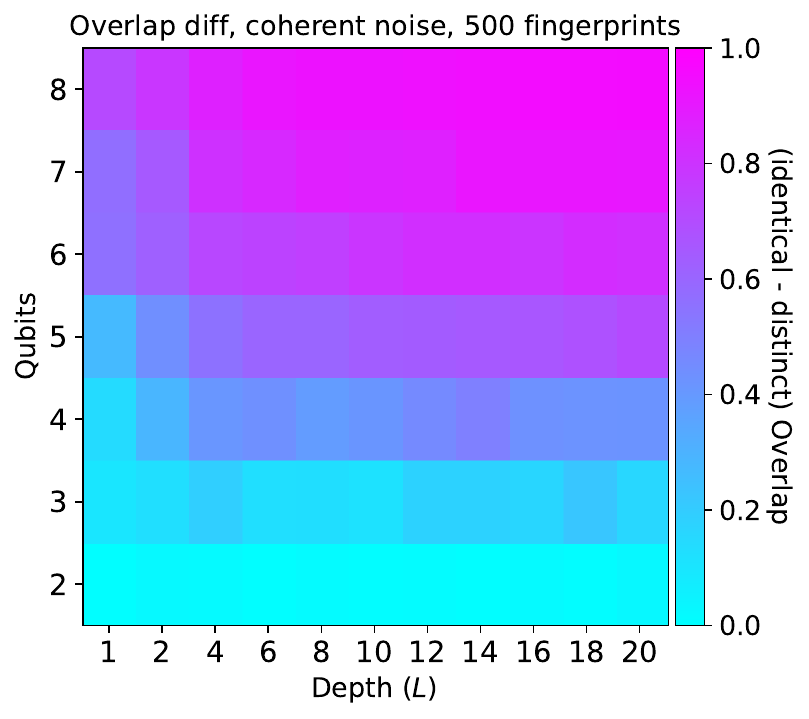}%
    }\hfill
    \subfloat[\label{subfig:noise-thermal}]{%
        \includegraphics[width=0.5\linewidth]{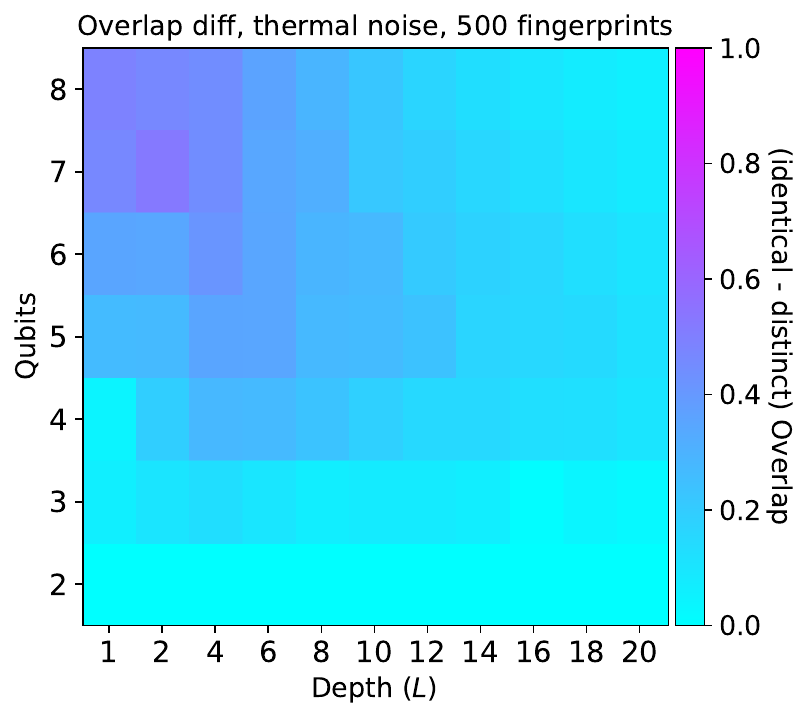}%
    }
    \subfloat[\label{subfig:noise-pauli}]{%
        \includegraphics[width=0.5\linewidth]{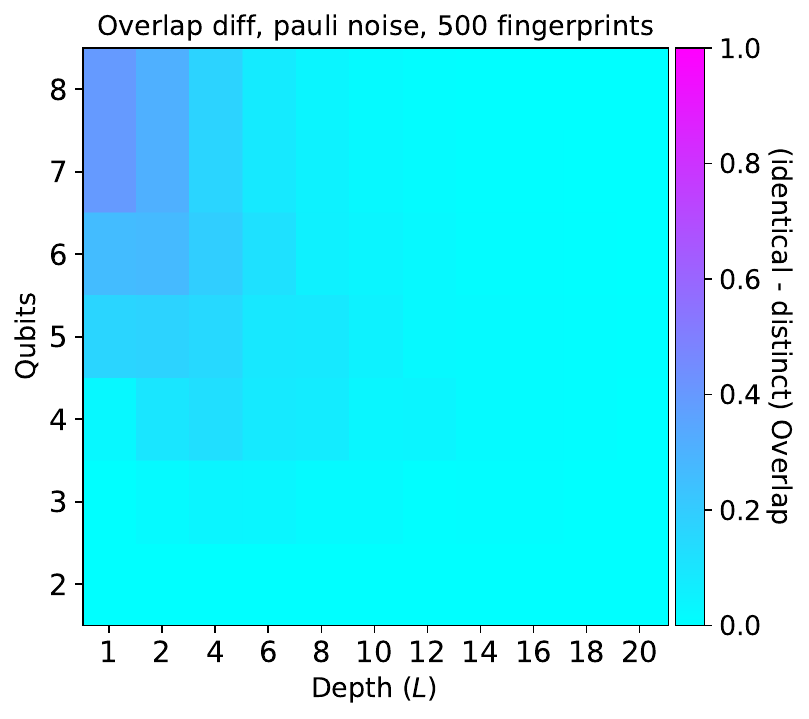}%
    }\hfill
    \caption{Difference between identical and distinct fingerprint overlaps obtained via (a) ideal noiseless simulation and noisy simulation using (b) coherent, (c) thermal, and (d) Pauli noise models, using the same pseudorandom circuit model as in Fig.~\ref{fig:noiseless_sims}(a). The Supercheq-EE protocol is robust to constant, coherent over rotations and even at short circuit depths it is able to successfully distinguish between a large number of inputs greater than the total number of individual basis states.}
    \label{fig:noisy-heatmap}
\end{figure}

We now consider the effect of noise, which can be a major barrier to the implementation of useful quantum circuits, including Supercheq. As shown in Fig.~\ref{fig:noiseless_sims}, in the noiseless case, it is beneficial to increase the depth of the random circuit to increase the amount of scrambling. However, with the consideration of noise, increasing circuit depth can adversely affect the Supercheq-EE protocol, for example by driving distinct random circuits towards the same maximally-mixed state. We thus investigate the sensitivity of Supercheq-EE's performance under the influence of different noise models including Pauli, thermal relaxation, and coherent noise channels.

We consider three separate noise models constructed using the Qiskit software package \cite{Qiskit}. Many quantum errors are often modeled as stochastic applications of Pauli matrices \cite{terhal2015quantum, erhard2019characterizing}.
We use a Pauli noise model which randomly applies an $X, Y, \text{ or } Z$ operation to any qubits participating in a single- or two-qubit gate with probabilities $p_X=p_Z=0.001$ and $p_Y=0.003$ based on the recent benchmarking of a superconducting quantum computer \cite{chen2022learnability}.
We also construct a thermal noise model which assigns an execution time to each gate---particularly, single-qubit rotations taking 100\,ns while two-qubit gates require 300\,ns---and a characteristic $T_1=50\,\mu\text{s}, T_2=70\,\mu\text{s}$ time to each qubit.
Finally, we also consider a deterministic coherent noise channel that applies a $\frac{\pi}{24}$ over rotation to any qubit participating in a single- or two-qubit gate. We summarize the performance of Supercheq-EE (using the protocol described in Fig.~\ref{fig:noiseless_sims}(a)) under these noise models at a variety of depths and qubit counts in Fig.~\ref{fig:noisy-heatmap}. For each pair of qubit and depth values, we report $(\omega_{i} - \omega_d)$, where $\omega_i$ is the minimum overlap observed between all pairs of \textit{identical} fingerprints, and $\omega_d$ is the maximum overlap observed between all pairs of \textit{distinct} fingerprints. Appendix~\ref{app:noisy} contains additional data from our noisy simulations including effects of increasing the total number of fingerprints evaluated.

As shown in Fig. \ref{fig:noiseless_sims}(b,c), deeper circuits will be necessary to generate sufficiently distant quantum states as the size of the input file grows and, therefore, the size of the files that Supercheq may be applied to will mainly be limited by noise. The noiseless and coherent noise model results in Fig.~\ref{fig:noisy-heatmap}(a,b) support this intuition, showing that the difference in overlaps grows as the depth of the circuits is increased. However, even for the noisier results in Fig.~\ref{fig:noisy-heatmap}(b,c), we see that even with relatively shallow circuits, Supercheq-EE is able to adequately distinguish between a total number of fingerprints (500) much greater than the number of individual basis states ($2^8$).
Noise mitigation techniques~\cite{PhysRevX.11.031057,PhysRevX.11.041036,errormit, ravi2022vaqem} as well as hardware improvements will be valuable to increase the size of files which can be verified using Supercheq, but ultimately, error correction and fault tolerance will provide a path to running Supercheq at scale.

\section{Experimental Results} \label{sec:experimental}

We experimentally realized proof-of-concepts underpinning Supercheq on both IBM superconducting hardware as well as Diraq superconducting hardware. In both cases, we demonstrate key primitives on physical qubits, as a stepping stone to future scalable error-corrected realizations of Supercheq.

\subsection{Supercheq-EE on IBM Superconducting Hardware}

To experimentally validate Supercheq-EE, we conducted fingerprinting on $n=3$-qubit states. Per Table III in \cite{scott2005optimal}, the best possible classical strategy with 3-qubit fingerprints would have a one-sided worst-case error of at least 50\% when encoding 9 states (meanwhile, encoding $2^3 = 8$ states would trivially have zero error). Meanwhile, ideal simulation results indicate that Supercheq-EE at shallow depth can achieve a worst-case error less than 46\% (see Fig.~\ref{fig:simulation_and_experiment}(a)).

\begin{figure}
    \centering
    \includegraphics[width=0.5\textwidth]{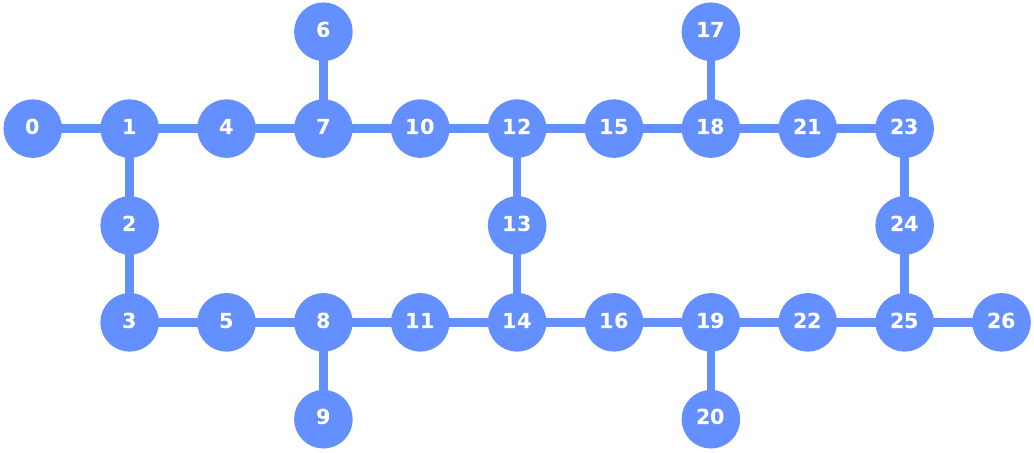}
    \caption{Topology of the 27-qubit ibm\_algiers backend. Our fingerprinting experiments were executed on 7- and 6- qubit segments for standard and destructive SWAP tests, respectively.\label{fig:algiers_topology}}
\end{figure}

This potential motivates our experiments, which were executed on IBM Quantum hardware accessed through Qiskit Runtime. In particular, we ran Supercheq-EE on the 27-qubit ibm\_algiers backend from the Falcon family (topology given in Fig.~\ref{fig:algiers_topology}), a generation of quantum computers that have achieved quantum volume of 512 and CLOPS (\emph{circuit layer operations per second}) of 15,000 \cite{gambetta2022expanding}. The high quantum volume is directly relevant to Supercheq because of the close relationship between quantum volume circuits \cite{cross2019validating} and the approximate unitary-$t$ designs relevant to Supercheq. Specifically, the quantum volume circuits are identical to the \textit{parallel random circuit} (with fully-connected topology) model for approximate unitary $t$-designs from \cite{brandao2016local}. Thus, hardware with high quantum volume is also able to produce larger approximate unitary $t$-design circuits. In addition, high CLOPS ensures that the quantum computer can rapidly produce a fingerprint state, as measured in wall clock time.

\begin{figure}
    \centering
    \includegraphics[width=0.5\textwidth]{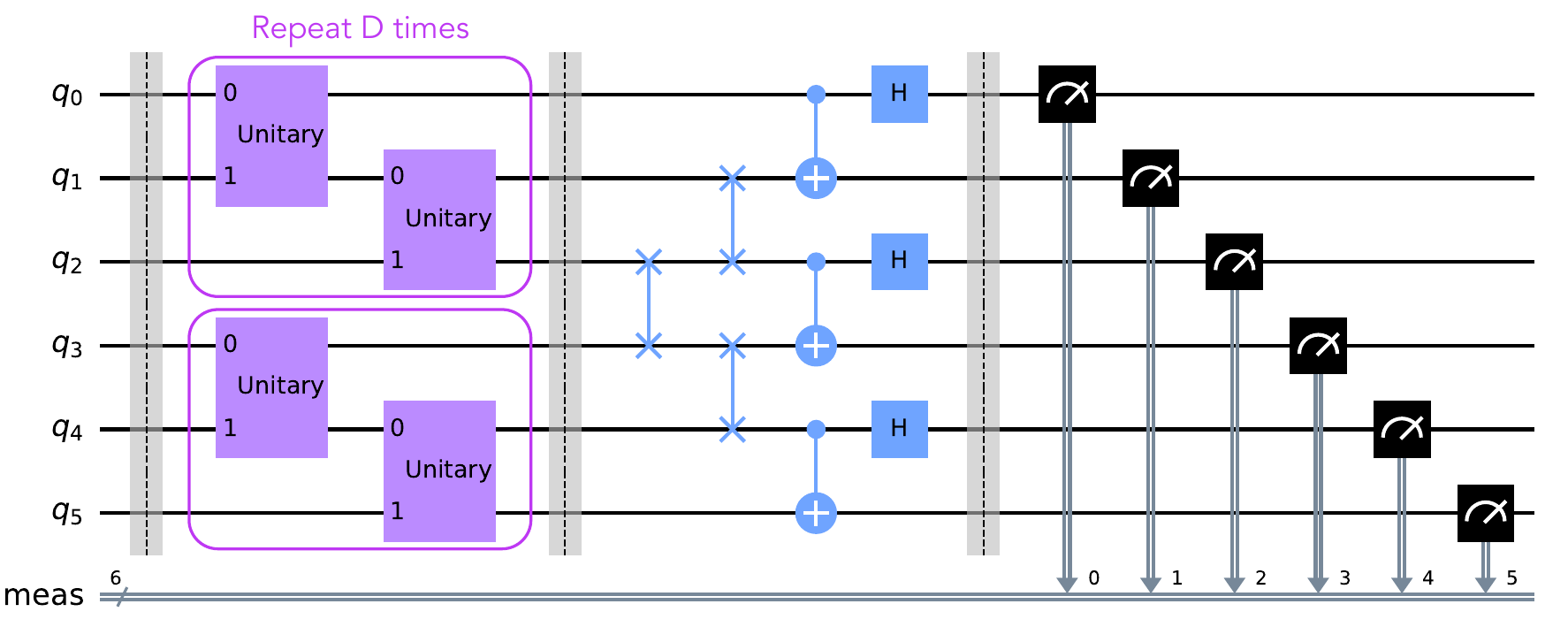}
    \caption{ Qiskit \cite{Qiskit} circuit schematic for Supercheq-EE with destructive SWAP test, which is well-suited to linear qubit topology.}
    \label{fig:circuit_schematic}
\end{figure}
Our experiments for Supercheq-EE involve two steps: an encoding circuit for Alice and Bob, followed by an inner product measurement step performed by the referee. The encoding circuit we use is the \textit{local circuit model} on a linear qubit topology, which is also proven to be an approximate unitary $t$-design in \cite{brandao2016local}. For three qubits arranged as $q_0\mbox{---}q_1\mbox{---}q_2$, the circuit is simply a sequence of Haar-random $\mathrm{SU}(4)$ unitaries, applied at random to either the $q_0\mbox{---}q_1$ or $q_1\mbox{---}q_2$ pair. This model again nearly coincides with Quantum Volume, only differing in that Quantum Volume would also allow $q_0\mbox{---}q_2$ gates (which would be recompiled to match the hardware topology).

\begin{figure}
    \centering
    \subfloat[]{%
        \includegraphics[width=0.25\textwidth]{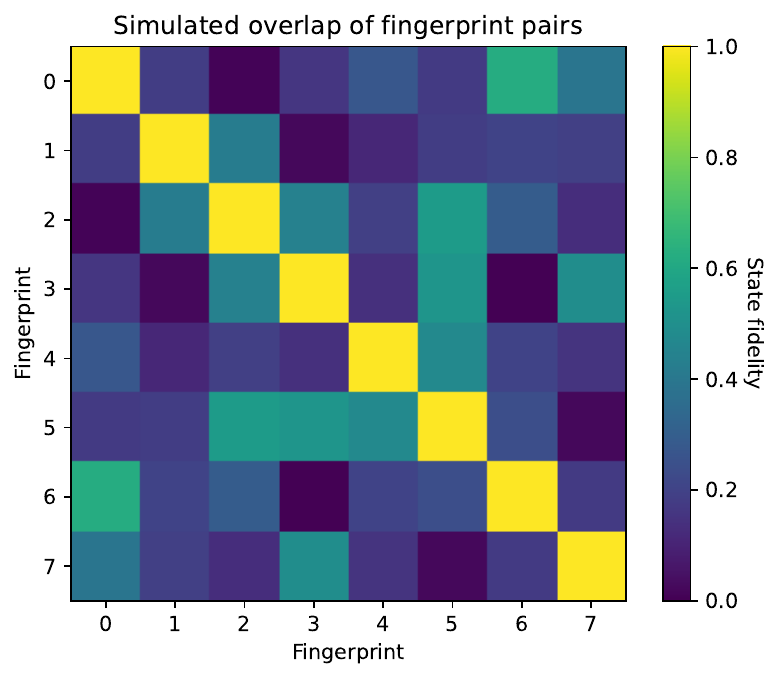}%
    }
    \subfloat[]{%
        \includegraphics[width=0.25\textwidth]{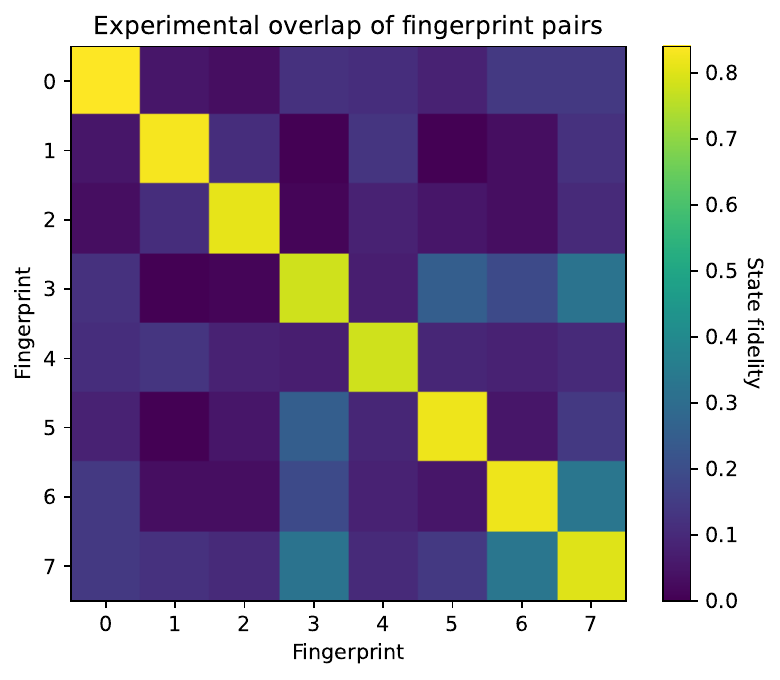}%
    }\\
    \subfloat[]{%
        \includegraphics[width=0.35\textwidth]{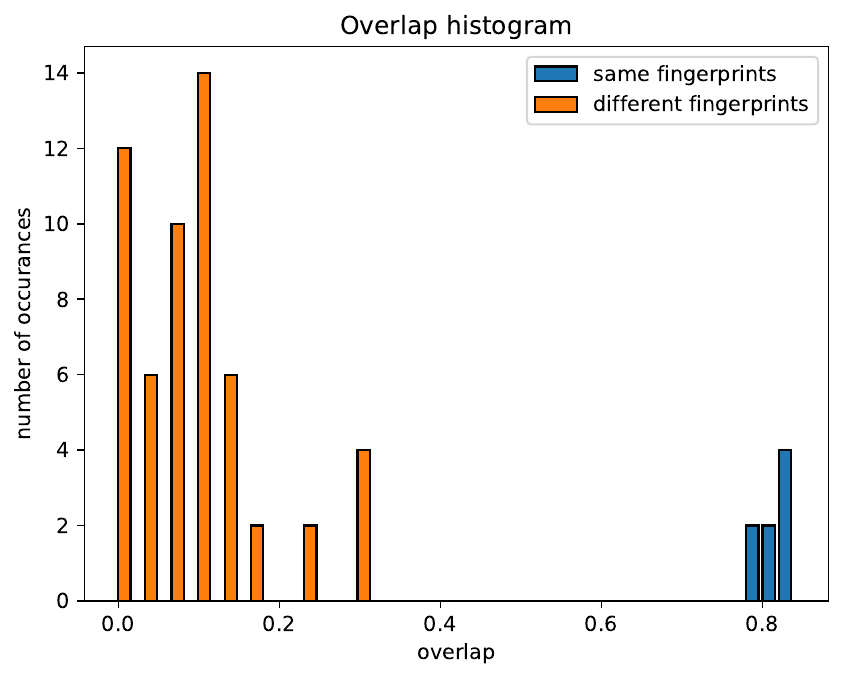}%
    }
    \caption{Overlap observed between 8 pairs of fingerprints, in (a) simulation and (b) experiment on ibm\_algiers using the destructive SWAP test averaged over $10^4$ shots. (c) Histogram of overlaps for the experiment on ibm\_algiers.}
    \label{fig:simulation_and_experiment}
\end{figure}
On limited-connectivity topologies, we found it advantageous to use the destructive SWAP test to distinguish fingerprints. In this variant, one performs a CX between every corresponding pair of qubits in Alice and Bob's fingerprints, followed by a Hadamard on every pair of Alice's fingerprints. Then, each qubit is measured. Classical post-processing on the measured $2n$ bits yields a fidelity estimate. On linear topology, the CX between every corresponding pair of qubits in Alice and Bob's fingerprints can be performed by a SWAP network which has shallow depth \cite{hashim2022optimized, tomesh2021coreset}, and each SWAP can be optimized to the native cross-resonance gates on IBM hardware \cite{gokhale2021faster}. The overall schematic of our Supercheq-EE experiments is shown in Fig.~\ref{fig:circuit_schematic}, including the destructive SWAP test. The SWAP network is the sequence of three SWAP gates; after these SWAPs, the qubits are arranged so that Alice and Bob's fingerprint qubits are interleaved. Thereafter, the three CX gates can be implemented with nearest-neighbor connectivity. The one disadvantage of this approach is that all qubits are measured at the end, so the fingerprints cannot be sent back to Alice and Bob to be reused.

The results---both from ideal simulation and from experiment---are shown in Fig.~\ref{fig:simulation_and_experiment}. In both cases, we apply a depth of just one layer as shown in Fig.~\ref{fig:circuit_schematic}. The heatmaps represent the symmetric matrix of fidelities between 9 input ``files". In the ideal simulation (Fig.~\ref{fig:simulation_and_experiment}(a)), the diagonal has 1.0 elements because every fingerprint has 1.0 fidelity with itself; off-diagonal, the max pairwise fidelity is 0.46. Experimentally (Fig.~\ref{fig:simulation_and_experiment}(b)) we found that all estimated fidelities (averaged over $10^4$ shots) are smaller, as would be expected in the limit of every state approaching the maximally mixed state due to noise. Note also that the destructive SWAP test is sensitive to measurement errors on $2n$ qubits, which is also an experimentally limiting factor.

However, a clear pattern persists on the diagonal, wherein we can still separate the case of identical fingerprints from the off-diagonal case of differing fingerprints. In particular, the minimum (self-)fidelity on the diagonal is 0.029, whereas the maximum off-diagonal fidelity is 0.017. The histogram of pairwise fidelities is shown in Fig.~\ref{fig:simulation_and_experiment}(c). Though the presence of noise here removes any quantum advantage compared to the best-known classical protocols, error-mitigation techniques~\cite{PhysRevX.11.031057,PhysRevX.11.041036,errormit, ravi2022vaqem} may sufficiently improve the performance of proof-of-principle tests of Supercheq on existing quantum devices to achieve better space complexity than known (asymptotically) optimal classical protocols. We leave further investigation along these lines for future work.

\subsection{Interference Circuit on Diraq Silicon Spin Qubit Hardware}

Our second experimental proof-of-concept measures overlaps between fingerprint states via the \textit{interference circuit} \cite{schuld2018supervised} which halves the total qubit count required relative to the SWAP test. In particular, recall that our goal is to estimate $\braket{\psi_A}{\psi_B}$. With the SWAP Test, we require two quantum data registers, $q_A$ and $q_B$, and one ancilla qubit. 
The inputs to the quantum registers $q_A$ and $q_B$ are $\ket{\psi_A}$ and $\ket{\psi_B}$, respectively. 

In comparison, the Interference circuit requires only one data register $q_d$ (plus one ancilla qubit). 
The difference is that, in the interference test, instead of inputting the quantum states directly into the circuit, 
we need to know the operations that create them—i.e., we need to know $U_A$ and $U_B$ such that
\begin{align}
    \ket{\psi_A} = U_A \ket{0}, \\
    \ket{\psi_B} = U_B \ket{0}.
\end{align}

\begin{figure}[t]
\centering
\begin{quantikz}
    \lstick{$\ket{0}$} & \gate{H} & \octrl{1} & \ctrl{1} & \gate{S^\lambda} & \gate{H} & \meter{} \\
    \lstick{$\ket{0}$} & \qw & \gate{U_A} & \gate{U_B} & \qw && \\
\end{quantikz}
\caption{Quantum circuit of Intereference-test}
\label{fig:intereference_circuit}
\end{figure}

The Interference-test circuit is depicted in Fig.~\ref{fig:intereference_circuit}. The probability of the measured ancilla qubit being zero, $P_0$, is 
\begin{align}
P_0 = 
\begin{cases}
    \frac{1}{2} \left( 1 + \Re \braket{\psi_A}{\psi_B} \right) & \lambda = 0, \\ \\
    \frac{1}{2} \left( 1 - \Im \braket{\psi_A}{\psi_B} \right) & \lambda = 1.
\end{cases}
\end{align}

To demonstrate a proof-of-concept example, we consider the data qubit register $q_d$ consisting of a single qubit. Our objective is to determine the state overlap between $\ket{0}$ and the set of states $\{\ket{0}, \ket{1}, \ket{+}, \ket{-}, \ket{i}, \ket{-i}\}$, representing the six cardinal points on the Bloch sphere. To this end, we set $U_A = I$ and $U_B \in \{I, X, H, HZ, HS, HS^\dagger\}$. This aligns with Supercheq implemented using one-qubit fingerprints. The interference circuit corresponding to this configuration (excluding the case $U_B = I$, which results in no gates after compilation) was implemented on Diraq’s two-qubit system. Fig.~\ref{fig:interference_result} presents the results. The agreement with the theoretical values—indicated by bold black dashed lines—demonstrates the validity of this proof-of-concept approach for evaluating the overlap between two quantum states.

The results exhibit dependence on the device temperature. As expected, lower temperatures yield closer agreement with the theoretical predictions. Moreover, the consistent trend observed in the estimated overlap values across different temperatures suggests that it may be possible to achieve the quantum advantage of Supercheq even at elevated temperatures, as the results reliably distinguish between non-identical quantum states. A more detailed investigation of this temperature dependence is left for future work.

\begin{figure}
\centering
\includegraphics[width=\columnwidth]{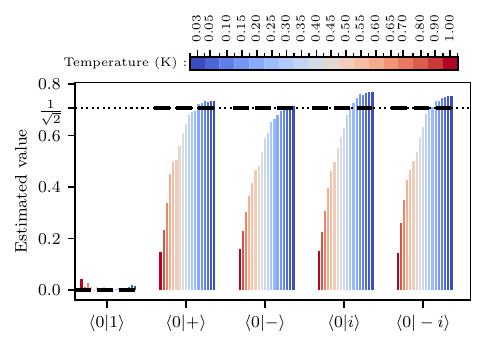}
\caption{Estimated state overlap between $\ket{0}$ and $\ket{\psi_B} \in \{\ket{1}, \ket{+}, \ket{-}, \ket{i}, \ket{-i}\}$, obtained using Diraq’s two-qubit device. The estimation depends on the device temperature. Bold dashed lines indicate the theoretical values, while the dotted line marks the value of $1/\sqrt{2}$.}
\label{fig:interference_result}
\end{figure}

\color{black}

\section{Conclusion}
\label{sec:conclusion}
In summary, Supercheq-EE gives prescriptive fingerprinting protocols that outperform classical approaches in the simultaneous message passing setting. While less efficient in compression than a previous quantum protocol~\cite{buhrman2001quantum}, Supercheq-EE scales gracefully to allow tradeoffs between encoding efficiency and quantum fingerprint preparation cost by invoking random circuit sampling. In effect, this endows quantum circuits from quantum supremacy and quantum volume experiments with a practical application.

Supercheq-IE complements this high-compression approach with a qualitatively different advantage: incrementality. In the graph subvariant, Supercheq-IE matches the optimal classical $\Theta(\sqrt{N})$ fingerprint scaling in the SMP setting while enabling constant-cost updates to the fingerprint under local changes to the underlying file. To our knowledge, no comparable non-cryptographic classical fingerprinting protocol achieves both of these properties. The hypergraph subvariant further strengthens this result, reducing the fingerprint size to $O(N^{1/\ell})$ for fixed $\ell$ while preserving constant-cost incremental updates. This makes Supercheq-IE particularly well suited to replicated, versioned, or frequently updated data, where maintaining and updating a fingerprint may be more valuable than generating one from scratch.

\subsection{Envisioned Applications}

While fingerprinting to check that two files are equal is a ubiquitous primitive throughout distributed data settings, the SMP model has certain limitations on its use. We begin by developing a characterization of a scenario well suited to Supercheq:
\begin{itemize}
    \item Data transmission is expensive. The cost of transmission can be a function of distance (e.g., transmission via satellites or to remote areas), data size (e.g., transmission of biological data, medical imagery, or government and military data), or logistical complexity (e.g., transmission of financial transactions across many tax jurisdictions).
    \item The two endpoints (Alice and Bob) cannot communicate directly and thus require  mediation through a third-party referee (e.g., secure multi-party computation and game theoretic applications, such as double auctions). In this scenario, SMP is already the necessary model.
\end{itemize}
Furthermore, Supercheq-IE is a natural fingerprinting protocol when data is replicated, versioned, or incrementally updated (e.g., data stored with a cloud provider).

We next conceptualize a few concrete categories of possible use cases:
\begin{itemize}
    \item Checking consistency of inputs. For instance, in distributed financial transactions, institutions want to ensure transactions are recorded at every node. Since transaction broadcasts can be lossy, Supercheq could be used to ensure that distributed parties received the same inputs. We note that the financial transaction consensus problem has high economic value. Another concrete application would be to RAID (Redundant Array of Inexpensive Disks) \cite{patterson1988case}; file storage redundancy could be provided by Supercheq rather than existing classical techniques.
    \item Checking consistency of processes. Even if we are assured that Alice and Bob have identical inputs, their processing techniques could be different. For example, if two nodes have different machine learning models, a third party may be interested in whether they produce the same outputs when provided with identical inputs. This use case maps well to VLSI design. In particular, minimizing communication within a chip is an important existing application for work on communication complexity \cite{mehlhorn1982vegas, kushilevitz1997communication}. In the VLSI context, Supercheq could be used to check if two spatially separated registers or program counters are identical; in this setting, Supercheq could conceivably help lower power consumption for CPUs and GPUs.
    \item Frequency checks for temporal updates. For example, in a distributed database application where Alice and Bob represent replicated nodes that are updated frequently, a referee may be interested in knowing the first instant at which inputs diverge. This is a use case for Supercheq-IE, since incrementality is critical.
    \item Alice (or Bob) is simply a trusted delegate of the refereee. For instance, suppose that the referee stores Bob's authentication credentials on the cloud with Alice. When Bob wishes to authenticate into the referee's access system, Bob can send a fingerprint of his credentials to the referee. When the referee tests against the (potentially) matching fingerprint from Alice, they can make a decision as to whether Bob has valid authentication tokens.
\end{itemize}

\subsection{Future Work}
We outline three thematic areas for continued investigation. First, a number of techniques could be used to improve the efficiency of Supercheq, both in terms of ``compression" and in terms of noise-resilience. To achieve higher compression, one could start by using qudits (systems with $d$ local dimensions, e.g. $d=3$ is a qutrit) for Supercheq-EE instead of qubits. In the Supercheq-IE setting, we also suggest exploring the use of qudit graph states~\cite{quditgraph} instead of qubit graph states. Another potential avenue for optimizing the efficiency of Supercheq may be by employing a number of pulse- or native gate- level techniques in the encoding and distinguishing circuits \cite{gokhale2019partial,shi2020resource, gokhale2020optimized}. One extreme case of this is a natural analog extension of Supercheq in which the random circuit is replaced by a scrambling analog quench or time-dependent Hamiltonian.

Second, we hope to initiate physical realization of Supercheq in a networked setting. In the local network case, linked quantum systems are emerging capable of coherently transmitting quantum states with modest fidelity between remote nodes containing multiple superconducting qubits~\cite{magnard2020microwave,zhong2021deterministic}. Furthermore, ion shuttling enables communication between spatially separated ion traps~\cite{pino2021demonstration}. Although quantum networking technology is in its early form, it lays the groundwork for what will eventually become infrastructure included in quantum-accelerated, warehouse-scale supercomputing clusters. 

Finally, we believe it would be fruitful to study connections to other quantum advantages in space or communication complexity. For example, recent work has demonstrated space advantages by using quantum computation for computing Boolean functions \cite{maslov2021quantum} and for linear regression \cite{montanaro2022quantum}. While quantum advantages in time complexity have been studied in depth, these advantages in space or communication complexity are relatively understudied and merit further investigation.

To kickstart further exploration of its potential, tutorials demonstrating Supercheq are now available in the Superstaq ~\cite{campbell2023superstaq} open-source repositories, \texttt{cirq-superstaq} (\href{https://github.com/SupertechLabs/cirq-superstaq/blob/main/examples/Supercheq.ipynb}{link}) and \texttt{qisit-superstaq} (\href{https://github.com/SupertechLabs/qiskit-superstaq/blob/main/examples/Supercheq.ipynb}{link}).

\section*{Acknowledgements}
This material was supported by the Australian Army and was featured in the Australian Army’s Quantum Technology Challenge (QTC) to prevent disruption of QCs and ensure their reliability and resilience. This material is supported in part by the U.S. Department of Energy, Office of Science, Office of Advanced Scientific Computing Research under Award No. DE-SC0021526 and DE-SC0025493. E.R.A.\ was partially supported by STAQ under award NSF Phy-1818914. We acknowledge the use of IBM Quantum Credits via the IBM Quantum Startups Program for this work. The views expressed are those of the authors and do not reflect the official policy or position of IBM or the IBM Quantum Platform team.

\bibliography{apssamp}

\appendix

\section{Proofs} \label{app:proofs}
Consider two randomly drawn $n$-qubit states, $\ket{\psi_i}$ and $\ket{\psi_j}$. We are interested in the fidelity (squared-overlap) between these two random states:

$$F_{ij} = \left\lvert\bra{\psi_i}\ket{\psi_j}\right\rvert^2.$$

\subsection{Haar-Random Sampling}

\Supercheqeethm*

\begin{proof}
In this case, the probability density for $F_{ij}$ when $\ket{\psi_i}$ and $\ket{\psi_j}$ are drawn from the Haar-random measure is \cite{zyczkowski2005average, kus1988universality}
$$p(f) = (2^n - 1)(1-f)^{2^n-2}.$$
Note that this is a $\text{Beta}(\alpha, \beta)$ distribution with $\alpha = 1, \beta = 2^n - 1$.

We are interested in upper-bounding the probability of the event
$$ \bigcup_{1 \leq i < j \leq K} F_{ij} > c$$
for some constant $c$. This event corresponds to $K$ randomly drawn states all being pairwise distinguishable by SWAP tests, with one-sided error decaying as $c^M$ when we perform $M$ SWAP tests. We will take $c=0.5$ for parity with the Supercheq-IE protocol, since differing graph states have fidelity of at most 0.5~\cite{PhysRevA.70.052328}. We calculate:

\begin{align}
\begin{split}
\Pr\left[ \bigcup_{1 \leq i < j \leq K} F_{ij} > 0.5 \right] &\leq \sum_{1 \leq i < j \leq K}\Pr[F_{ij} > 0.5] \\
& \text{\quad (Union Bound)} \\
&= \frac{K(K-1)}{2} \Pr[F_{12} > 0.5] \\
&= \frac{K(K-1)}{2} \left( \frac{1}{2} \right)^{2^n - 1} \\
& \text{\quad (CDF of Beta dist.)} \\
&= \frac{K(K-1)}{2^{2^n}}.
\end{split} \label{eq:haar_random_collision}
\end{align}

For $K \in o(\sqrt{2^{2^n}})$, for example $K = 1.4^{2^n}$, the probability of an indistinguishable pair approaches $0$ as $n \to \infty$. For example, even at just $n = 20$, this probability is upper bounded by $\sim 10^{-9000}$.
\end{proof}

\subsection{Approximate Unitary $t$-Design Sampling}

We now examine sampling approximate unitary $t$-designs. We first give a more formal definition of the notion of $t$-designs we consider here; namely, \emph{approximate $t$-designs under the monomial measure}~\cite{harrow2018approximate}, specializing here to qubits.
\begin{definition}[Monomial definition of approximate $t$-designs~\cite{harrow2018approximate}]
$\mu$ is a monomial-based $\epsilon$-approximate $t$-design on $n$ qubits if expectations of all degree-$\left(t,t\right)$ monomials are within additive error $\epsilon 2^{-nt}$ of those resulting from the Haar measure.
\end{definition}

Armed with this definition, we now prove Theorem~\ref{thm:t_design}. Unfortunately, as we require extremely tight errors---i.e., $\epsilon=\exp\left(-\operatorname{poly}\left(n\right)\right)$---we are outside of the regime of more efficient recent developments in implementing $t$-designs~\cite{10756150,doi:10.1126/science.adv8590,schuster2025strongrandomunitariesfast}, which require $\epsilon\geq\exp\left(-\operatorname{O}\left(n\right)\right)$.
\Supercheqeetdesthm*
\begin{proof}
    Let $U_i$ be i.i.d. drawn from an $\epsilon$-approximate unitary $t$-design under the monomial measure~\cite{harrow2018approximate}, and let $\ket{\psi_i}=U_i\ket{0}$. Let $\tilde{F}_{ij}$ be the random variable:
    \begin{equation}
        \tilde{F}_{ij}=\left\lvert\bra{\psi_i}\ket{\psi_j}\right\rvert^2,
    \end{equation}
    and $F_{ij}$ the same for Haar-random states. Since $F_{ij}\sim\text{Beta}(1, 2^n - 1)$, we have moments:
    \begin{equation}
        \E[F_{ij}^t] = \prod_{k=0}^{t-1} \frac{\alpha + k}{\alpha + \beta + k} = \prod_{k=1}^t \frac{k}{2^n + i - 1} \leq \frac{t^t}{2^{nt}}.
    \end{equation}
    As $U_i$ are i.i.d. drawn from an $\epsilon$-approximate unitary $t$-design under the monomial measure, and as the $t$th moment of the fidelity can be written as a sum of $2^{2nt}$ expectations of monomial terms (each with coefficient $1$), we thus have that:
    \begin{equation}
        \E[\tilde{F}_{ij}^t]\leq\E[F_{ij}^t]+\left\lvert\E[\tilde{F}_{ij}^t]-\E[F_{ij}^t]\right\rvert\leq \frac{t^t}{2^{nt}}+2^{nt}\epsilon.
    \end{equation}
    We therefore have by Markov's inequality that:
    \begin{equation}
        \Pr[\tilde{F}_{ij} > 0.5] \leq \frac{\E[\tilde{F}_{ij}^t]}{(0.5)^t} \leq 2^t\left(2^{nt}\epsilon+\frac{t^t}{2^{nt}}\right).
    \end{equation}
    Applying the Union Bound as in Eq.~\eqref{eq:haar_random_collision}, we see that:
    \begin{equation}
        \begin{aligned}
            \Pr&\left[ \bigcup_{1 \leq i < j \leq K} \tilde{F}_{ij} > 0.5 \right]\\
            &\leq \frac{K(K-1)}{2}\left(2^{t\left(n+1\right)}\epsilon+\frac{t^t}{2^{t\left(n-1\right)}}\right).
        \end{aligned}
    \end{equation}
    We now take $t=n^{\ell-1}$ and $\epsilon=2^{-2n^\ell}$. Note by Corollary 1 of \cite{Haferkamp2022randomquantum} for sufficiently large $n$, there exist nearest-neighbor random quantum circuits in 1D that sample from this distribution in depth $O\left(t^{4+o\left(1\right)}\left(nt+\log\left(\epsilon^{-1}\right)\right)\right)=O\left(n^{5.01\ell-4.01}\right)$. We then have that:
    \begin{equation}
        \begin{aligned}
            \Pr&\left[ \bigcup_{1 \leq i < j \leq K} \tilde{F}_{ij} > 0.5 \right] \\
            &\leq\frac{K(K-1)}{2}\left(1+2^{\left(\ell-1\right)n^{\ell-1}\log_2\left(n\right)}\right)2^{-n^{\ell-1}\left(n-1\right)}.
        \end{aligned}
    \end{equation}
    Let $K=1.4^{n^\ell}$. Then,
    \begin{equation}
        \begin{aligned}
            \Pr&\left[ \bigcup_{1 \leq i < j \leq K} \tilde{F}_{ij} > 0.5 \right]\\
            &=O\left(2^{-\left(1-2\log_2\left(1.4\right)+o\left(1\right)\right)n^\ell}\right)\to 0.
        \end{aligned}
    \end{equation}
\end{proof}

\section{Additional Noisy Simulation Results}
\label{app:noisy}

We here present additional numerical results to supplement those discussed in Sec.~\ref{subsec:noisy_sim}. Fig.~\ref{fig:noisy-max-overlap} measures the maximum fidelity between all distinct fingerprints produced via the Supercheq-EE protocol under the noise models discussed in Sec.~\ref{subsec:noisy_sim}; it is easy to see there is some threshold depth at which the performance of the protocol is optimal at a given noise rate. In Fig.~\ref{fig:noisy-heatmap-appendix}, we measure the fidelity between Supercheq-EE fingerprints for $5$-bit seed files and a depth $5$ fingerprinting circuit under the noise models discussed in Sec.~\ref{subsec:noisy_sim}. A clear distinction in state fidelities exists between fingerprints corresponding to identical seeds and those corresponding to distinct seeds.

\begin{figure*}
    \centering
    \subfloat[]{%
        \includegraphics[width=0.4\textwidth]{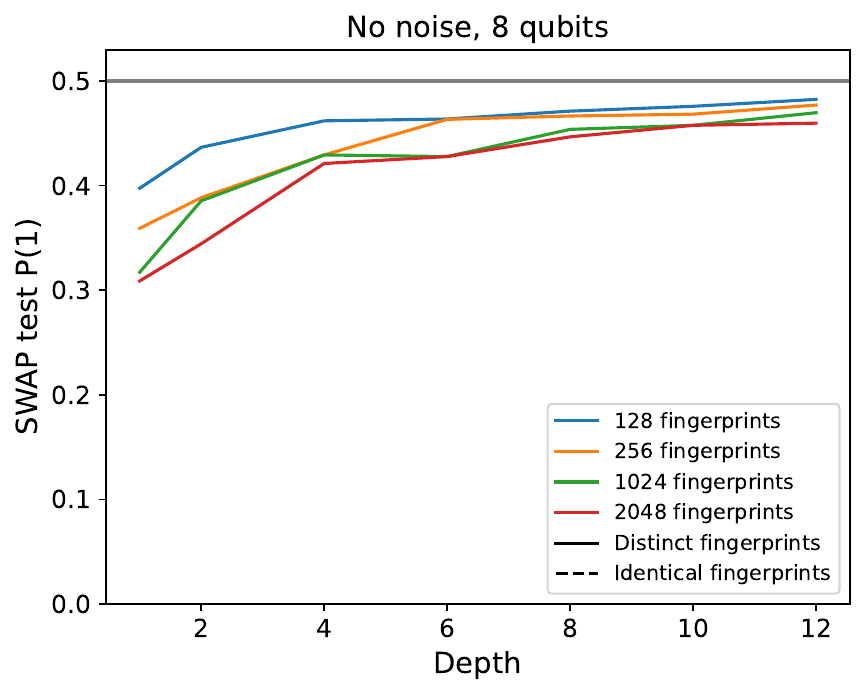}%
    }
    \subfloat[]{%
        \includegraphics[width=0.4\textwidth]{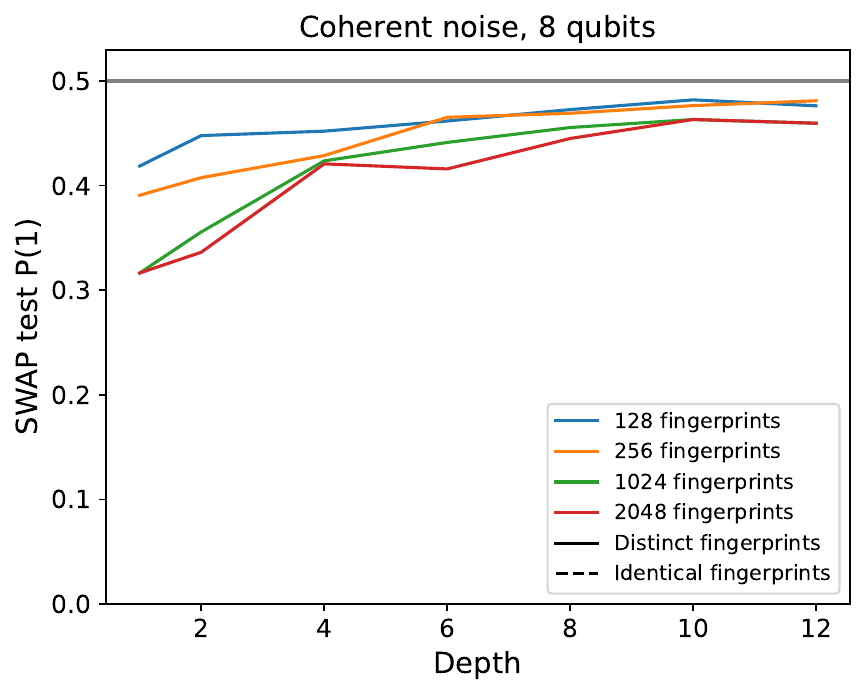}%
    }\hfill
    \subfloat[]{%
        \includegraphics[width=0.4\textwidth]{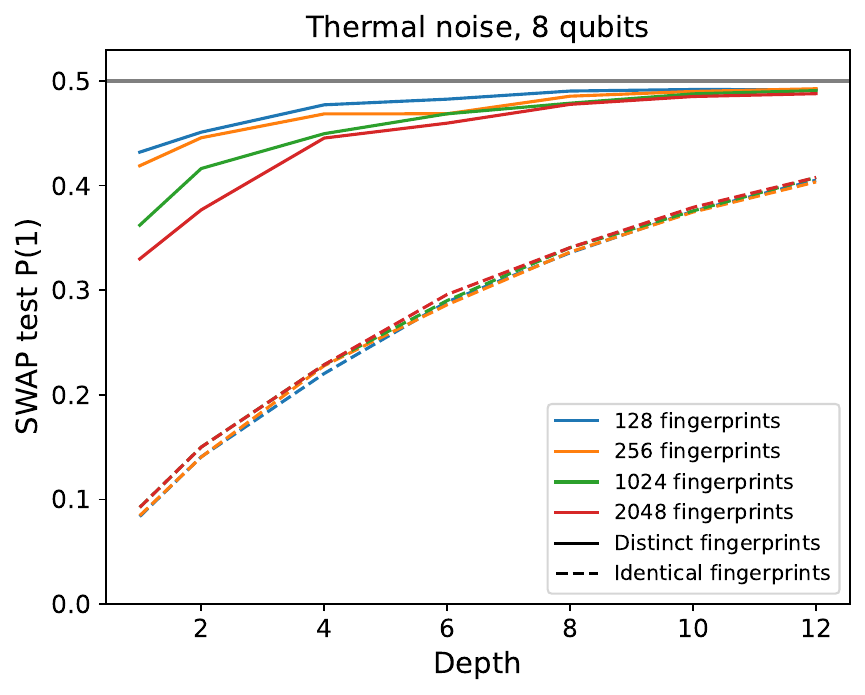}%
    }
    \subfloat[]{%
        \includegraphics[width=0.4\textwidth]{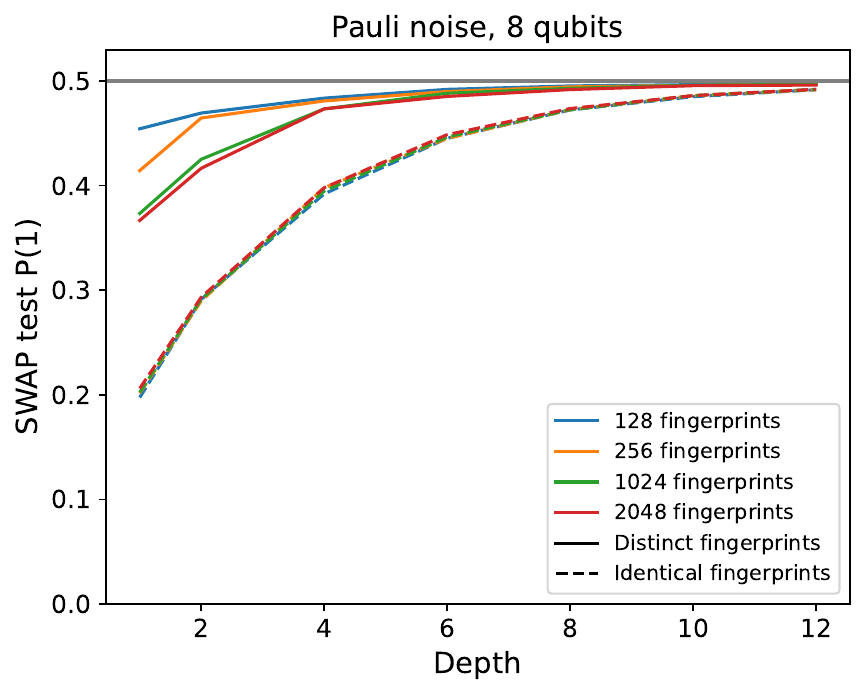}%
    }
    \caption{The y-axis shows the probability of measuring the $\ket{1}$ state when performing the SWAP test between two Supercheq fingerprints. For an increasing total number of fingerprints tested, we consider the SWAP test performed between distinct (solid lines) and identical (dashed lines) fingerprints under (a) noiseless, (b) coherent, (c) thermal, and (d) Pauli noise models. As the depth of the Supercheq-EE protocol is increased, the probability of measuring the ``1'' outcome of the SWAP test converges to 50\% for both the distinct and identical fingerprint pairs. Notably, the coherent noise model considered here, which applies a constant over rotation to every operation, has no effect on the outcome of the SWAP test.}
    \label{fig:noisy-max-overlap}
\end{figure*}

\begin{figure*}
    \centering
    \subfloat[\label{subfig:pauli-overlap}]{%
        \includegraphics[width=0.33\textwidth]{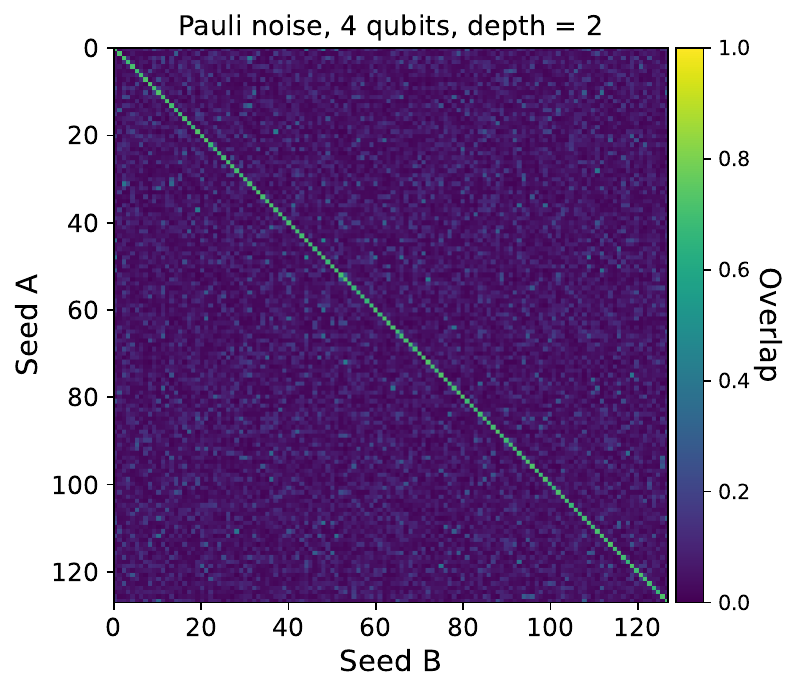}%
    }\hfill
    \subfloat[\label{subfig:thermal-overlap}]{%
        \includegraphics[width=0.33\textwidth]{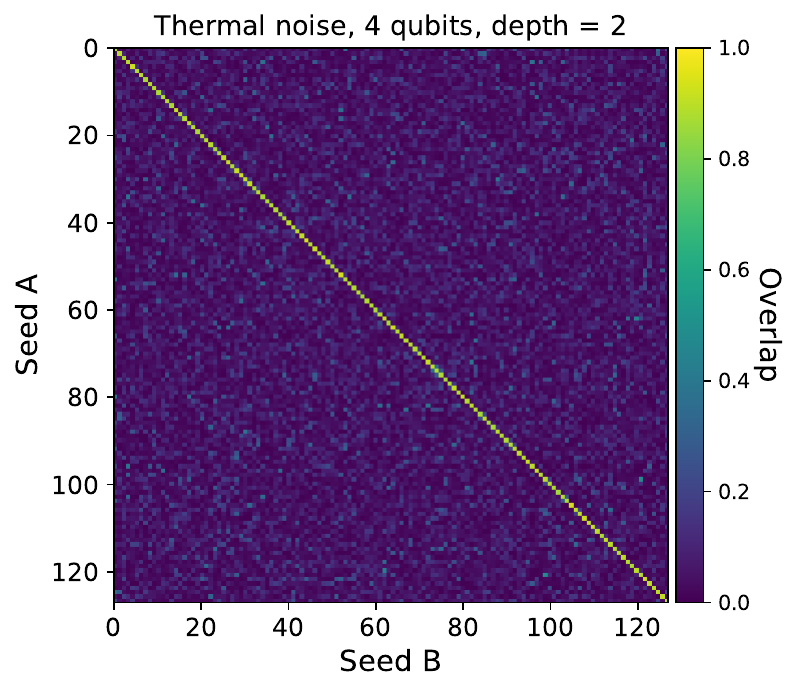}%
    }\hfill
    \subfloat[\label{subfig:coherent-overlap}]{%
        \includegraphics[width=0.33\textwidth]{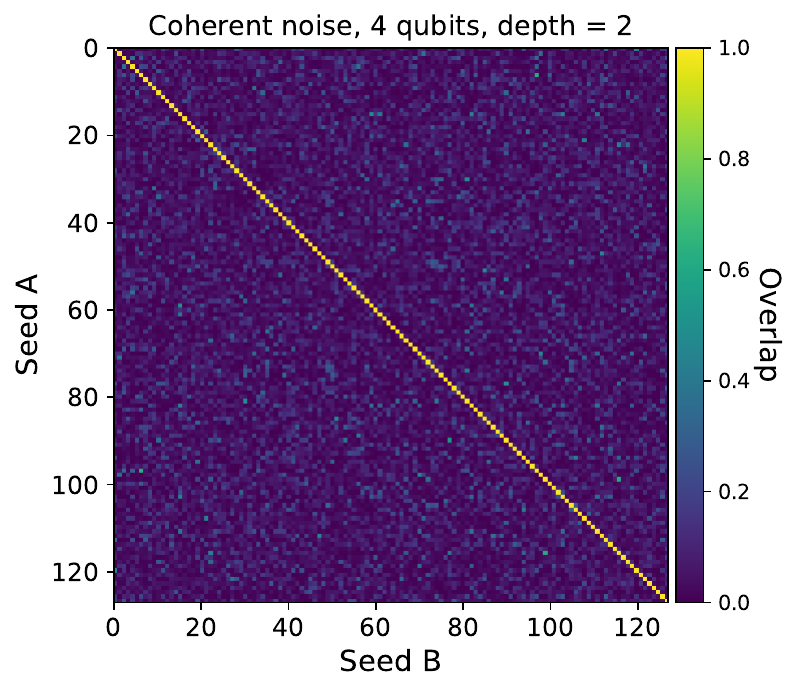}%
    }
    \caption{Overlap between all pairs of 128 files converted into Supercheq fingerprints with 4 qubits and a depth of 5 under (a) Pauli, (b) thermal, and (c) coherent noise models. Even at low depths, the fingerprints are easily distinguishable.}
    \label{fig:noisy-heatmap-appendix}
\end{figure*}

\end{document}